\documentclass[journal,onecolumn]{IEEEtran}
\usepackage{amssymb}
\usepackage{amsmath}
\usepackage{longtable}
\newtheorem{theorem}{Theorem}[section]
\newtheorem{lemma}[theorem]{Lemma}
\newtheorem{proposition}[theorem]{Proposition}
\newtheorem{corollary}[theorem]{Corollary}

\newtheorem{remark}{Remark}

\usepackage[para]{threeparttable}

\newcommand{\C}{{\mathcal C}}
\newcommand{\F}{{\mathbb F}}
\newcommand{\f}{{\mathbb F}}

\newenvironment{proof}[1][Proof]{\begin{trivlist}
\item[\hskip \labelsep {\bfseries #1}]}{\end{trivlist}}

\makeatletter

\newcommand{\Rmnum}[1]{\expandafter\@slowromancap\romannumeral #1@}
\makeatother

\ifCLASSINFOpdf

\else

\fi

\begin{document}
%
\title{New characterization and parametrization of LCD Codes}

\author{ Claude Carlet$^1$ \and Sihem Mesnager$^2$ \and Chunming Tang$^3$ \and  Yanfeng Qi$^4$
\thanks{This work was supported by SECODE project and
the National Natural Science Foundation of China
(Grant No. 11401480, 11531002, 11701129). C. Tang
also acknowledges support from 14E013,
CXTD2014-4 and the Meritocracy Research Funds  of China West Normal University.
Y. Qi also acknowledges support from Zhejiang provincial Natural Science Foundation of China (LQ17A010008 and LQ16A010005).
}

\thanks{C. Carlet is with the Department of Mathematics, University of Paris VIII, 93526 Saint-Denis, France and also with University of Paris XIII, CNRS, LAGA UMR 7539,  Sorbonne Paris Cit\'e, 93430 Villetaneuse, France (email: claude.carlet@univ-paris8.fr).}
\thanks{S. Mesnager is with the Department of Mathematics, University of Paris VIII, 93526 Saint-Denis, France, also with University of Paris XIII, CNRS, LAGA UMR 7539, Sorbonne Paris Cit\'e, 93430 Villetaneuse, France, and also with Telecom ParisTech 75013 Paris (email: smesnager@univ-paris8.fr).}
\thanks{C. Tang is with School of Mathematics and Information, China West Normal University, Nanchong, Sichuan,  637002, China. e-mail: tangchunmingmath@163.com
}

\thanks{Y. Qi is with School of Science, Hangzhou Dianzi University, Hangzhou, Zhejiang, 310018, China.
e-mail: qiyanfeng07@163.com
}

}


\maketitle

\begin{abstract}

Linear complementary dual (LCD) cyclic codes were referred historically to as reversible cyclic codes, which had applications in data storage. Due to a newly discovered application in cryptography, there has been renewed interest in LCD codes.
In particular, it has been shown that  binary LCD codes
 play an important role in  implementations against side-channel attacks and fault injection attacks.
 In this paper, we first present a new characterization of binary LCD codes in terms of  their symplectic basis. Using such a characterization,we solve a conjecture proposed by Galvez et al.  on the minimum distance of binary LCD codes. Next, we consider the action
 of the orthogonal group on the set of all LCD codes, determine all possible orbits of this action, derive simple closed  formulas of the
 size of the orbits, and present some asymptotic results of the size of the corresponding orbits.
 Our results show that almost all binary LCD codes are odd-like codes with odd-like duals,
 and  about half of $q$-ary LCD codes have orthonormal  basis, where $q$ is a power of an odd prime.
\end{abstract}


%
\IEEEpeerreviewmaketitle

\section{Introduction}
Let $q$ be a power of a prime.  $\F_q$ and $\F_q^n$ denote the finite field with $q$ elements and $n$-dimensional vector space over $\F_q$
respectively. Let $\mathrm{wt}(\mathbf x)$  denote the weight of $\mathbf x \in \F_q^n$, i.e., the number of nonzero elements in $\mathbf x=(x_1, x_2,\ldots , x_n)$. For any $\mathbf x=(x_1, x_2,\ldots , x_n)$ and $\mathbf y=(y_1, y_2,\ldots , y_n)$ in $\F_q^n$,  the
\emph{Euclidean inner product} of $\mathbf x$ and $\mathbf y$  is defined by
\begin{align*}
\mathbf x \cdot \mathbf y=\sum_{i=1}^n x_i y_i.
\end{align*}
An $[n, k, d]$ linear code $\C$ over $\f_q$ is a $k$-dimensional subspace
of $\F_q^n$ with minimum (Hamming) distance $d$.  The \emph{dual} of $\C$
is defined by
\begin{align*}
\C^\perp=\{\mathbf w \in \f_q^n: \mathbf w \cdot \mathbf c =0 \text{ for all } \mathbf c \in \C\}.
\end{align*}
If $\C \cap \C^\perp =\{\mathbf{0}\}$, then $\C$ is called a \emph{linear complementary dual code} or an \emph{LCD code}.

Carlet and Guilley applied LCD codes in side-channel attacks (SCA) and fault non-invasive attacks \cite{BCCGM14,CG14}. A lot of works has been devoted to construct LCD codes.  In \cite{DLL16}, Ding et al. constructed several families of Euclidean LCD cyclic codes over finite fields and analyzed their parameters.  In \cite{LDL16} Li et al. studied two special families of LCD cyclic codes, which are both BCH codes.
Mesnager et al. \cite{MTQ16} presented a  construction of algebraic geometry Euclidean LCD codes. In \cite{CMTQ17},
Carlet et al. completely determined all $q$-ary ($q > 3$) Euclidean LCD codes and all $q^2$-ary ($q > 2$) Hermitian LCD codes for all possible parameters.
In the latest paper, Carlet et al. \cite{CMTQ17-sigma},  introduced the concept of linear codes with
 $\sigma$ complementary dual ($\sigma$-LCD), which includes known Euclidean LCD codes, Hermitian LCD codes, and Galois LCD codes. Their results extend those on the classical LCD codes and 
 show that $\sigma$-LCD codes allow the construction of  LCP of codes more easily and with more flexibility.

However, little is known on  the general structure  of LCD codes.
The goal of this paper is to characterize LCD codes and  to study the structure of the set  of LCD codes.  We first present a new characterization of LCD codes in terms of  their symplectic basis. Such a characterization,
 allows us to solve a conjecture proposed by Galvez et al. \cite{GKLRW17} on the minimum distance of binary LCD codes. Afterwards, we consider the action
 of the orthogonal group on the set of all LCD codes, determine all possible orbits of this action, give simple closed formulas of the
 size of the orbits, and present some asymptotic results of the size of the corresponding orbits.
 Our results show that almost all binary LCD codes are odd-like codes with odd-like duals,
 and  about half of $q$-ary LCD codes have orthogonal basis, where $q$ is a power of an odd prime.

The paper is organized as follows.
In Section \ref{sec-pre}, we recall some basic results on LCD codes and binary symmetric matrices.
In Section \ref{sec:char-2}, we firstly characterize binary LCD codes in terms of their basis. Based on these results we solve a conjecture proposed by Galvez et al. \cite{GKLRW17}.
In Section \ref{sec:action-2}, we consider the action of orthogonal group on the set of all binary LCD codes and study the orbits of this action.
In Section \ref{sec:action-odd}, we use the same method to characterize $q$-ary LCD codes, where $q$ is a power of an odd prime.

\section{Preliminaries}\label{sec-pre}
For a matrix $G$, $G^T$ denotes the transposed matrix of $G$. The following  characterization of LCD codes is due to Massey \cite{Mas92}.
\begin{proposition}\label{prop:LCD-Gram}
 Let $\C$ be a linear code with a generator matrix $G$ and a parity-check matrix $H$. Then the three following properties are equivalent:

(i) $\C$ is LCD;

(ii) the matrix $GG^T$ is invertible;

(iii) the matrix $HH^T$ is invertible.
\end{proposition}

A matrix $M$ is symmetric if $M^T=M$.
A diagonal matrix is a matrix in which the entries outside the main diagonal are all zero.
We shall write $\mathrm{diag}[x_1, ..., x_k]$ for a diagonal matrix whose main diagonal entries  are $x_1, ..., x_k$.
The following proposition gives the classification of symmetric matrices over $\F_2$ under the equivalence relation $M \sim QMQ^T$  \cite{Ser12,Weyl16},
where  $Q$  is  nonsingular.
\begin{proposition}\label{prop:symmetic-matix}
Let $M$ be a symmetric $k \times k$ matrix of rank $t$ with entries in $\F_2$.

(i) If all diagonal entries of $M$ equal $0$,  then $t$ is even and there is a nonsingular matrix $Q$ such that
\begin{align*}
QMQ^T=\mathrm{diag}[
\underbrace{J_2, J_2, \ldots, J_2}_t, 0, \ldots, 0],
\end{align*}
where $J_2=
  \left [ \begin{array}{cc}
      0 & 1  \\
      1 & 0
    \end{array} \right ]$.

(ii) If at least one diagonal entry of $M$ is nonzero, then there is a nonsingular $k \times  k$ matrix $Q$ such that
\begin{align*}
QMQ^T=\mathrm{diag}[
\underbrace{1, 1, \ldots, 1}_{t}, 0, \ldots, 0].
\end{align*}
\end{proposition}

\section{New characterization of binary LCD codes by their basis}\label{sec:char-2}
In this section, we will present a new characterization of binary LCD codes. Based on this characterization, we
prove a conjecture on minimum distance of binary LCD codes proposed by Galvez et al. \cite{GKLRW17}.

 A vector $\mathbf x=(x_1, x_2, \ldots, x_n)$ in $\F_2^n$ is even-like if
\begin{align*}
\sum_{i=1}^n  x_i=0
\end{align*}
and is odd-like otherwise. A code is said to be even-like if it has only even-like codewords, and is said to be odd-like if it is not even-like.

\begin{theorem}\label{thm:odd-basis}
Let $\C$ be an odd-like binary code with parameters $[n, k]$. Then $\C$ is LCD if and only if there exists a  basis $\mathbf c_1, \mathbf c_2, \ldots, \mathbf c_k$ of  $\C$ such that
for any $i ,j \in \{1,2, \ldots, k\}$,
$\mathbf c_i \cdot \mathbf c_j$ equals $1$ if $i=j$ and equals $0$ if $i\neq j$.
\end{theorem}
\begin{proof}
If there exists a  basis $\mathbf c_1, \mathbf c_2, \ldots, \mathbf c_k$ of  $\C$  such that
for any $i ,j \in \{1,2, \ldots, k\}$,
$\mathbf c_i \cdot \mathbf c_j$ equals $1$ if $i=j$ and equals $0$ if $i\neq j$. Let $G$ be the generator matrix of $\C$ that corresponds to the base $\mathbf c_1, \mathbf c_2, \ldots, \mathbf c_k$. Then, $GG^T=I_k$, where $I_k$ is the $k \times k$ identity matrix. From Proposition \ref{prop:LCD-Gram}, $\C$ is LCD.

Conversely, assume that $\C$ is LCD. Let $G'$ be a generator matrix of $\C$. Then from Proposition \ref{prop:LCD-Gram}, $G'G'^T$ is an invertible symmetric matrix of size $k \times k$.
Since $\C$ is odd-like, there is at least one nonzero  diagonal entry on $G'G'^T$. Then from Part (ii) of Proposition \ref{prop:symmetic-matix} there exists nonsingular $k\times k$ matrix
 $Q$ such that $QG'G'^TQ^T=(QG')(QG')^T= I_k$, where $I_k$ is the $k \times k$ identity matrix. Let $G=QG'$. Then, $G$ is also a generator matrix of $\C$. Let $\mathbf c_i$
 be the $i$-th row of matrix $G$ for $i\in \{1,2, \ldots, k\}$. Hence $\mathbf c_1,\mathbf c_2, \ldots, \mathbf c_k$ is the desired basis, since $GG^T=I_k$. It completes the proof.
\end{proof}

\begin{remark}
Theorem \ref{thm:odd-basis} shows that a binary odd-like code $\C$ is LCD if and only if $\C$ has an orthonormal basis.
\end{remark}

\begin{theorem}\label{thm:even-basis}
Let $\C$ be an even-like binary code with parameters $[n, k]$. Then $\C$ is LCD if and only if  $k$ is even and there exists a  basis $\mathbf c_1, \mathbf c_1', \mathbf c_2, \mathbf c_2', \ldots, \mathbf c_{\frac{k}{2}}, \mathbf c_{\frac{k}{2}}'$ of  $\C$  such that
for any $i ,j \in \{1,2, \ldots, \frac{k}{2}\}$,  the following conditions hold

(i) $\mathbf c_i \cdot \mathbf c_j= \mathbf c_i' \cdot \mathbf c_j'=0$;

(ii) $\mathbf c_i \cdot \mathbf c_j'=0$ if $i\neq j$;

(iii) $\mathbf c_i \cdot \mathbf c_i'=1$.
\end{theorem}
\begin{proof}
If $k$ is even and there exists a  basis $\mathbf c_1, \mathbf c_1', \mathbf c_2, \mathbf c_2', \ldots, \mathbf c_{\frac{k}{2}}, \mathbf c_{\frac{k}{2}}'$ of  $\C$  such that
for any $i ,j \in \{1,2, \ldots, \frac{k}{2}\}$,  the following conditions hold

(i) $\mathbf c_i \cdot \mathbf c_j= \mathbf c_i' \cdot \mathbf c_j'=0$;

(ii) $\mathbf c_i \cdot \mathbf c_j'=0$ if $i\neq j$;

(iii) $\mathbf c_i \cdot \mathbf c_i'=1$.

 Let $G$ be the generator matrix of $\C$ that corresponds to the base $\mathbf c_1, \mathbf c_1', \mathbf c_2, \mathbf c_2', \ldots, \mathbf c_{\frac{k}{2}}, \mathbf c_{\frac{k}{2}}'$.
 Then, $GG^T=\mathrm{diag}[
\underbrace{J_2, J_2, \ldots, J_2}_{\frac{k}{2}}]$, where $J_2=\left [ \begin{matrix}
0 & 1\\
1 & 0
\end{matrix} \right ]
$. Hence, $GG^T$ is invertible. From Proposition \ref{prop:LCD-Gram}, we have that $\C$ is LCD.

Conversely, assume that $\C$ is LCD. Let $G'$ be a generator matrix of $\C$. Then from Proposition \ref{prop:LCD-Gram}, $G'G'^T$ is an invertible symmetric matrix of size $k \times k$.
Since $\C$ is even-like, all diagonal entries of  $G'G'^T$ equal $0$. Then from Part (i) of Proposition \ref{prop:symmetic-matix}, $k$ is even and there exists nonsingular $k\times k$ matrix
 $Q$ such that $QG'G'^TQ^T=(QG')(QG')^T= \mathrm{diag}[
\underbrace{J_2, J_2, \ldots, J_2}_{\frac{k}{2}}]$, where $J_2=\left [ \begin{matrix}
0 & 1\\
1 & 0
\end{matrix} \right ]
$. Let $G=QG'$. Then $G$ is also a generator matrix of $\C$. Let $\mathbf c_i$
 be the $(2i-1)$-th row of matrix $G$ and $\mathbf c_i'$
 be the $2i$-th row of matrix $G$  for $i\in \{1,2, \ldots, \frac{k}{2}\}$. Hence $\mathbf c_1, \mathbf c_1', \mathbf c_2, \mathbf c_2', \ldots, \mathbf c_{\frac{k}{2}}, \mathbf c_{\frac{k}{2}}'$ is the  desired basis, because $GG^T=\mathrm{diag}[
\underbrace{J_2, J_2, \ldots, J_2}_{\frac{k}{2}}]$. It completes the proof.
\end{proof}

\begin{remark}
A basis of $\C$  satisfying the conditions (i), (ii), and (iii) in Theorem \ref{thm:even-basis} is called a \emph{symplectic basis}. Theorem \ref{thm:even-basis} shows that an even-like code is LCD if and only if
it has a  symplectic basis.
\end{remark}

Let $d_{\mathrm{LCD}}(n,k)$ be the maximum of possible values of $d$ among $[n, k, d]$ binary LCD codes. Dougherty et al. \cite{DKOSS17}
 gave a linear programming bound on the largest size of   $d_{\mathrm{LCD}}(n,k)$. Using the idea of principal submatrices, Galvez et al. \cite{GKLRW17}
 proved that $d_{\mathrm{LCD}}(n,k)\le d_{\mathrm{LCD}}(n,k-1)$ if $k$ is odd and $d_{\mathrm{LCD}}(n,k)\le d_{\mathrm{LCD}}(n,k-2)$ if $k$ is even.
 They also conjectured that $d_{\mathrm{LCD}}(n,k)\le d_{\mathrm{LCD}}(n,k-1)$ for any $k$. We will prove this conjecture using the new characterization of binary LCD codes described above.
 To this end,  the following lemma is needed.
\begin{lemma}\label{lem:1-1}
 Let $\C$ be an even-like binary LCD code with parameters $[n, k]$. Then, there exists a  basis $\mathbf c_1, \mathbf c_1', \mathbf c_2, \mathbf c_2', \ldots, \mathbf c_{\frac{k}{2}}, \mathbf c_{\frac{k}{2}}'$ of  $\C$  such that
for any $i ,j \in \{1,2, \ldots, \frac{k}{2}\}$,  the following conditions hold

(i) $\mathbf c_i \cdot \mathbf c_j= \mathbf c_i' \cdot \mathbf c_j'=0$;

(ii) $\mathbf c_i \cdot \mathbf c_j'=0$ if $i\neq j$;

(iii) $\mathbf c_i \cdot \mathbf c_i'=1$;

(iv) $c_{i,1}=c_{i,1}'$, where $\mathbf c_i=(c_{i,1}, \ldots, c_{i,n})$ and $\mathbf c_i'=(c_{i,1}', \ldots, c_{i,n}')$.

 \end{lemma}

\begin{proof}
From Theorem \ref{thm:even-basis}, there exists a  basis
 $\mathbf c_1, \mathbf c_1', \mathbf c_2, \mathbf c_2', \ldots, \mathbf c_{\frac{k}{2}}, \mathbf c_{\frac{k}{2}}'$ of  $\C$, which satisfy the conditions (i), (ii) and (iii).

Without loss of generality, assume that $c_{i,1}=c'_{i,1}+1=1$ for $1\le i\le l$ and
$c_{i,1}=c'_{i,1}$ for $l+1\le i\le k$, where $l$ is some positive integer, $\mathbf c_i=(c_{i,1}, \ldots, c_{i,n})$, and $\mathbf c_i'=(c_{i,1}', \ldots, c_{i,n}')$.
Let $\mathbf w_i=\mathbf c_i$ for $i\in \{1, \ldots, \frac{k}{2}\}$, $\mathbf w_i'=\mathbf c_i'+\mathbf{c}_i$  for $i\in \{1, \ldots, l\}$ and
$\mathbf w_i'=\mathbf c_i'$ for $i\in \{l+1, \ldots, \frac{k}{2}\}$.
It can be verified directly that the basis $\mathbf w_1, \mathbf w_1', \mathbf w_2, \mathbf w_2', \ldots, \mathbf w_{\frac{k}{2}}, \mathbf w_{\frac{k}{2}}'$
satisfy the conditions (i), (ii),  (iii) and (iv).
\end{proof}

\begin{theorem}\label{thm:d-d}
If $2 \le k \le  n$, then $d_{\mathrm{LCD}}(n,k)\le d_{\mathrm{LCD}}(n,k-1)$.
\end{theorem}
\begin{proof}
Let $k$ be odd and $\C$ be an $[n,k]$ LCD code with minimum distance $d=d_{\mathrm{LCD}}(n,k)$. Then, $\C$ is odd-like. From Theorem \ref{thm:odd-basis}, there exists a basis $\mathbf c_1,
\ldots, \mathbf c_{k}$ of $\C$ such that $\mathbf c_i\cdot \mathbf c_j=1$ if $i=j$ and $\mathbf c_i\cdot \mathbf c_j=0$ otherwise. From Theorem \ref{thm:odd-basis}, the code $\C'=\mathrm{Span}
\{\mathbf c_1,
\ldots, \mathbf c_{k-1}\}$ is an $[n, k-1]$ LCD code with minimum distance at least $d_{\mathrm{LCD}}(n,k)$.

Let $k$ be even and $\C$ be an $[n,k]$ LCD code  with minimum distance $d=d_{\mathrm{LCD}}(n,k)$. If $\C$ is odd-like, the results follow from a similar discussion in the case $k$ odd.

In the following, assume that $\C$ is even-like. If for any $\mathbf c=(c_1, \ldots, c_n) \in \C$, $c_1=0$. Let $\mathbf c_1, \mathbf c_1',\ldots, \mathbf c_{\frac{k}{2}}, \mathbf c_{\frac{k}{2}}'$ be a basis of $\C$ satisfying
the conditions in Theorem \ref{thm:even-basis} and
$\C'=\mathrm{Span}\{\mathbf c_1, \mathbf c_1',\ldots, \mathbf c_{\frac{k}{2}-1}, \mathbf c_{\frac{k}{2}-1}', \mathbf w_{\frac{k}{2}}\}$, where $\mathbf w_{\frac{k}{2}}= \mathbf c_{\frac{k}{2}} +\mathbf e_1$. Then, $\C'$ is an $[n,k-1]$ code with minimum distance at least $d_{\mathrm{LCD}}(n,k)$. Let $G'$ be the generator matrix corresponding to the basis $\mathbf c_1, \mathbf c_1',\ldots, \mathbf c_{\frac{k}{2}-1}, \mathbf c_{\frac{k}{2}-1}', \mathbf w_{\frac{k}{2}}$. Then, $G'G'^T=\mathrm{diag}[
\underbrace{J_2, J_2, \ldots, J_2}_{\frac{k}{2}-1}, 1]$, where $J_2=\left [ \begin{matrix}
0 & 1\\
1 & 0
\end{matrix} \right ]
$. It is observed that $G'G'^T$ is nonsingular. By Proposition \ref{prop:LCD-Gram}, $\C'$ is LCD. Thus, the results hold.

If for some $\mathbf c=(c_1, \ldots, c_n) \in \C$, $c_1\neq 0$. Then, there exists a  basis
 $\mathbf c_1, \mathbf c_1', \mathbf c_2, \mathbf c_2', \ldots, \mathbf c_{\frac{k}{2}}, \mathbf c_{\frac{k}{2}}'$ of  $\C$ satisfying the conditions in Lemma \ref{lem:1-1}.
 Without  loss of generality, assume that $c_{i,1}=c_{i,1}'=1$ for $i\in \{1, \ldots, l\}$ and $c_{i,1}=c_{i,1}'=0$ for $i\in \{l+1, \ldots, \frac{k}{2}\}$, where $l$ is some positive integer.
 Let $\C'=
 \mathrm{Span}\{\mathbf c_1' +\mathbf c_1 +\mathbf e_1\} +\mathrm{Span}\{ \mathbf c_i +\mathbf c_1, \mathbf c_i' +\mathbf c_1: 2\le i \le l\}
 + \mathrm{Span}\{ \mathbf c_i, \mathbf c_i' : l+1\le i \le \frac{k}{2}\}$.
 Then, $\C'$ is an $[n,k-1]$ code with minimum distance at least $d_{\mathrm{LCD}}(n,k)$. Let $G'$ be the generator matrix corresponding to the basis$\{\mathbf c_1' +\mathbf c_1 +\mathbf e_1\}
  \cup \{ \mathbf c_i +\mathbf c_1, \mathbf c_i' +\mathbf c_1: 2\le i \le l\}  \cup \{ \mathbf c_i, \mathbf c_i' : l+1\le i \le \frac{k}{2}\}$.
 It is observed  that, for $2\le i \le l$,
 \begin{align*}
 (\mathbf c_1'+ \mathbf c_1+ \mathbf e_1)\cdot (\mathbf c_i+ \mathbf c_1)=(\mathbf c_1'+ \mathbf c_1+ \mathbf e_1)\cdot (\mathbf c_i'+ \mathbf c_1)=1,
 \end{align*}
for $l\le i \le \frac{k}{2}$,
 \begin{align*}
 (\mathbf c_1'+ \mathbf c_1+ \mathbf e_1)\cdot \mathbf c_i=(\mathbf c_1'+ \mathbf c_1+ \mathbf e_1)\cdot \mathbf c_i'=0,
 \end{align*}
 and $(\mathbf c_1'+ \mathbf c_1+ \mathbf e_1)\cdot (\mathbf c_1'+ \mathbf c_1+ \mathbf e_1)=1$. Then,
 \begin{align*}
 G'G'^T=
 \left [ \begin{matrix}
 1 & \mathbf u\\
 \mathbf u^T & \mathrm{diag}[
\underbrace{J_2, J_2, \ldots, J_2}_{\frac{k}{2}-1}]
 \end{matrix} \right ],
 \end{align*}
 where $\mathbf u=(\underbrace{1, \ldots, 1}_{2(l-1)}, 0, \ldots, 0)\in \F_2^{k-2}$.
That is, $G'G'^T$ is a matrix of the following  form
\begin{align}\label{eq:G-GT}
\left [
\begin{array}{ccccccccccc}
1 & 1 & 1 &\ldots & 1& 1 & 0 &0 \ldots 0 & 0\\
1 & 0 & 1 &\ldots & 0& 0 & 0 &0 \ldots 0 & 0\\
1 & 1 & 0 &\ldots & 0& 0 & 0 &0 \ldots 0 & 0\\
\vdots & \vdots & \vdots &\ldots & \vdots & \vdots & \vdots &\vdots \ldots \vdots & \vdots\\
1 & 0 & 0 &\ldots & 0& 1 & 0 &0 \ldots 0 & 0\\
1 & 0 & 0 &\ldots & 1 & 0 & 0 &0 \ldots 0 & 0\\
0 & 0 & 0 &\ldots & 0& 0 & 0 &1 \ldots 0 & 0\\
0 & 0 & 0 &\ldots & 0 & 0 & 1 &0 \ldots 0 & 0\\
\vdots & \vdots & \vdots &\ldots & \vdots & \vdots & \vdots &\vdots \ldots \vdots & \vdots\\
0 & 0 & 0 &\ldots & 0& 0 & 0 &0 \ldots 0 & 1\\
0 & 0 & 0 &\ldots & 0 & 0 & 0 &0 \ldots 1 & 0
\end{array}
\right]
\end{align}

Adding the $2$-th column,  $3$-th column, ...,$(2l-1)$-th column of Matrix (\ref{eq:G-GT}) to  the first column of Matrix (\ref{eq:G-GT}), one has
\begin{align*}
\left [
\begin{array}{ccccccccccc}
1 & 1 & 1 &\ldots & 1& 1 & 0 &0 \ldots 0 & 0\\
0 & 0 & 1 &\ldots & 0& 0 & 0 &0 \ldots 0 & 0\\
0 & 1 & 0 &\ldots & 0& 0 & 0 &0 \ldots 0 & 0\\
\vdots & \vdots & \vdots &\ldots & \vdots & \vdots & \vdots &\vdots \ldots \vdots & \vdots\\
0 & 0 & 0 &\ldots & 0& 1 & 0 &0 \ldots 0 & 0\\
0 & 0 & 0 &\ldots & 1 & 0 & 0 &0 \ldots 0 & 0\\
0 & 0 & 0 &\ldots & 0& 0 & 0 &1 \ldots 0 & 0\\
0 & 0 & 0 &\ldots & 0 & 0 & 1 &0 \ldots 0 & 0\\
\vdots & \vdots & \vdots &\ldots & \vdots & \vdots & \vdots &\vdots \ldots \vdots & \vdots\\
0 & 0 & 0 &\ldots & 0& 0 & 0 &0 \ldots 0 & 1\\
0 & 0 & 0 &\ldots & 0 & 0 & 0 &0 \ldots 1 & 0
\end{array}
\right],
\end{align*}
which is nonsingular. From Proposition \ref{prop:LCD-Gram}, $\C'$ is LCD, which completes the proof.

\end{proof}

\section{The Action of the orthogonal group on the set of binary LCD codes}\label{sec:action-2}
In this section, we will consider the action of the  orthogonal group on the set of binary LCD codes.

The set of all binary LCD codes with parameters $[n,k]$ is denoted by $\mathrm{LCD}[n,k]$.
Let $\mathrm{LCD}_{o,o}[n,k]$ ($\mathrm{LCD}_{o,e}[n,k]$, respectively) be the set of all odd-like binary LCD codes $\C$ of length $n$ and dimension $k$ such that  $\C^\perp$ is odd-like
 (even-like, respectively). We can define $\mathrm{LCD}_{e,o}[n,k]$ and $\mathrm{LCD}_{e,e}[n,k]$  similarly.
Obviously there is no even-like binary LCD code with even-like dual and $\mathrm{LCD}_{e,e}[n,k]$ is an empty set. Thus,
  \begin{align}\label{eq:LCD-o-e}
  \mathrm{LCD}[n,k]=\mathrm{LCD}_{o,o}[n,k] \cup \mathrm{LCD}_{o,e}[n,k] \cup \mathrm{LCD}_{e,o}[n,k].
  \end{align}
  Further, the mapping $\C \rightarrow \C^{\perp}$ gives a one-to-one correspondence between $\mathrm{LCD}_{o,e}[n,k]$ and $\mathrm{LCD}_{e,o}[n,n-k]$.
  We  shall say that two binary LCD codes $\C$ and $\C'$ have the same type if they are both in  $\mathrm{LCD}_{o,o}[n,k]$, or $\mathrm{LCD}_{o,e}[n,k]$, or $\mathrm{LCD}_{e,o}[n,k]$.

For any $\mathbf{v}_1, \ldots, \mathbf{v}_k\in \F_2^n$, $\mathrm{Span}\{\mathbf{v}_1, \ldots, \mathbf{v}_k\}$ denotes the linear subspace of $\F_2^n$ spanned by $\mathbf{v}_1, \ldots, \mathbf{v}_k$. Let $\mathbf e_i=(0, \ldots, 1, \ldots, 0)\in \F_2^n$ be the vector with  $1$ in the $i$-th position. We define the following three matrices:
\begin{align}\label{eq:G-o-o}
G_{o,o}=\left [\begin{matrix}
\mathbf{e}_1\\
\mathbf{e}_2\\
\mathbf{e}_3\\
\mathbf{e}_4\\
\vdots\\
\mathbf{e}_{k-1}\\
\mathbf{e}_{k}
\end{matrix} \right ], G_{o,e}=\left [\begin{matrix}
\sum_{i=1}^{n-k+1} \mathbf{e}_i\\
\mathbf{e}_{n-k+2}\\
\mathbf{e}_{n-k+3}\\
\mathbf{e}_{n-k+4}\\
\vdots\\
\mathbf{e}_{n-1}\\
\mathbf{e}_{n}
\end{matrix} \right ] \text{ and }G_{e,o}=\left [\begin{matrix}
\mathbf{e}_1+\mathbf{e}_{2}\\
\mathbf{e}_{2}+\mathbf{e}_{3}\\
 \mathbf{e}_1+\mathbf{e}_2+\mathbf{e}_3+\mathbf{e}_{4}\\
\mathbf{e}_4+\mathbf{e}_{5}\\
\vdots\\
\sum_{i=1}^{k-1} \mathbf{e}_{i}+\mathbf{e}_{k}\\
\mathbf{e}_{k}+\mathbf{e}_{k+1}
\end{matrix} \right ].
\end{align}
Let $\C_{o,o}, \C_{o,e}$ and $\C_{e,o}$ denote the linear codes generated by the  matrices  $G_{o,o}$, $G_{o,e}$ and $G_{e,o}~(k~\text{even})$, respectively.

Let $0< k <n$. From Theorem \ref{thm:even-basis}, if $k$ ($n-k$, respectively) is odd, $\mathrm{LCD}_{e,o}[n,k]$ ($\mathrm{LCD}_{o,e}[n,k]$, respectively) is empty.
The following proposition shows that if $k$ ($n-k$, respectively) is even, $\mathrm{LCD}_{e,o}[n,k]$ ($\mathrm{LCD}_{o,e}[n,k]$, respectively) is non empty.
\begin{proposition}\label{prop:C-o-e}
Let $0<k<n$. Let $\C_{o,o}, \C_{o,e}$ and $\C_{e,o}$ be defined as above. Then, $\C_{o,o}\in \mathrm{LCD}_{o,o}[n,k]$. Moreover, if $k$ ($n-k$, respectively) is even,
$\C_{e,o}\in \mathrm{LCD}_{e,o}[n,k]$ ($\C_{o,e}\in \mathrm{LCD}_{o,e}[n,k]$, respectively).
\end{proposition}
\begin{proof}
It is easy to verify that $\C_{o,o}^{\perp}= \mathrm{Span}\{\mathbf{e}_i: i\in\{ k+1, \ldots, n \}\}$. Then, $\C_{o,o} \cap \C_{o,o}^{\perp} =\{\mathbf{0}\}$ and $\C_{o,o}\in \mathrm{LCD}_{o,o}[n,k]$.

It is observed that
\begin{align*}
\C_{e,o}^{\perp}= \mathrm{Span}\{\sum_{i=1}^{k+1} \mathbf{e}_{i}, \mathbf{e}_{k+2}, \ldots, \mathbf{e}_{n}\},
\end{align*}
and
\begin{align*}
\C_{o,e}^{\perp}= \mathrm{Span}\{ \mathbf{e}_i +\mathbf{e}_{n-k+1}: i\in \{1,2, \ldots, n-k\}\}.
\end{align*}
Then, if $k$ ($n-k$, respectively) is even, $\C_{e,o} \cap \C_{e,o}^{\perp}=\{\mathbf{0}\}$ ($\C_{e,o} \cap \C_{e,o}^{\perp}=\{\mathbf{0}\}$, respectively), which completes the proof.
\end{proof}

\begin{remark}
Let $H_{o,o}$, $H_{o,e}$ and $H_{e,o}$ be matrices defined by
\begin{align}\label{eq:H-o-o}
H_{o,o}=\left [\begin{matrix}
\mathbf{e}_{k+1}\\
\mathbf{e}_{k+2}\\
\mathbf{e}_{k+3}\\
\mathbf{e}_{k+4}\\
\vdots\\
\mathbf{e}_{n-1}\\
\mathbf{e}_{n}
\end{matrix} \right ], H_{o,e}=\left [\begin{matrix}
\sum_{i=1}^{1}\mathbf{e}_i+\mathbf{e}_{2}\\
\mathbf{e}_2+\mathbf{e}_{3}\\
\sum_{i=1}^{3}\mathbf{e}_i+\mathbf{e}_{4}\\
\mathbf{e}_4+\mathbf{e}_{5}\\
\vdots\\
\sum_{i=1}^{n-k-1} \mathbf{e}_{i}+\mathbf{e}_{n-k}\\
\mathbf{e}_{n-k}+\mathbf{e}_{n-k+1}
\end{matrix} \right ] \text{ and }
H_{e,o}=\left [\begin{matrix}
\sum_{i=1}^{k+1} \mathbf{e}_i\\
\mathbf{e}_{k+2}\\
\mathbf{e}_{k+3}\\
\mathbf{e}_{k+4}\\
\vdots\\
\mathbf{e}_{n-1}\\
\mathbf{e}_{n}
\end{matrix} \right ].
\end{align}
\end{remark}
Then, from the proof of Proposition \ref{prop:C-o-e}, $H_{o,o}$, $H_{o,e}$ and $H_{e,o}$ are parity-check matrices of $\C_{o,o}$, $\C_{o,e}$ and $\C_{e,o}$ respectively.

For any $[n, k]$ linear code $\C$ and  $n\times n$ matrix $Q$, let $\C Q$ be the linear code defined by
\begin{align}\label{eq:C-Q}
\C Q=\{\mathbf c Q: \mathbf c \in \C\}.
\end{align}

In the rest of the paper, $\mathbb{GL}_n$ denotes  the general linear group  of degree $n$ over $\F_2$, which
 is the set of $n\times n$ invertible matrices over $\F_2$, together with the operation of ordinary matrix multiplication.
 A binary orthogonal matrix or orthogonal matrix is a square matrix with binary entries whose columns and rows are orthogonal unit vectors (i.e., orthonormal vectors), i.e.
$Q^{T}Q=QQ^{T }=I$, where $I$ is the identity matrix. The set of $n \times n$ orthogonal matrices forms a group $\mathbb{O}_n$, known as the orthogonal group.
Recall that an $n\times n$ matrix $Q$ is an orthogonal matrix if and only if $(\mathbf u Q)\cdot (\mathbf v Q)= \mathbf u \cdot \mathbf v$ for any $\mathbf u, \mathbf v \in \F_2^n$.
\begin{theorem}\label{thm:C-Q}
Let $\C_1$ and $\C_2$ be two binary $[n,k]$ LCD codes of the same type. Then, there exists an orthogonal matrix $Q\in \mathbb{O}_n$ such that $\C_2=\C_1 Q$.
Conversely, for any binary $[n,k]$ LCD code $\C$ and any orthogonal matrix $Q\in \mathbb{O}_n$, $\C Q$ is also an  LCD code with the same type as $\C$.
\end{theorem}
\begin{proof}
We first consider the case $\C_1, \C_2\in \mathrm{LCD}_{o,o}[n,k]$. Then, $\C_1$ and $\C_1^{\perp}$ are odd-like LCD codes.
From Theorem \ref{thm:odd-basis}, there is an orthonormal basis $\mathbf c_1, \ldots, \mathbf c_k$ of $\C_1$ and  an orthonormal basis $\mathbf c_{k+1}, \ldots, \mathbf c_{n}$ of $\C_1^{\perp}$.
Then, $\mathbf c_1, \ldots, \mathbf c_n$ is an orthonormal basis of $C_1 \oplus C_1^{\perp}=\F_2^n$. Similarly, there is another orthonormal basis $\mathbf w_1, \ldots, \mathbf w_n$ of $\F_2^n$ such that $\C_2=\mathrm{Span}\{\mathbf w_1, \ldots, \mathbf w_k\}$ and $\C_2^\perp=\mathrm{Span}\{\mathbf w_{k+1}, \ldots, \mathbf w_{n}\}$.
Let $Q_1$ and $Q_2$ be orthogonal matrices defined by
\begin{align*}
Q_1=\left [ \begin{matrix}
\mathbf c_1\\
\mathbf c_2\\
\vdots \\
\mathbf c_n
\end{matrix} \right ] \text{ and } Q_2=\left [ \begin{matrix}
\mathbf w_1\\
\mathbf w_2\\
\vdots \\
\mathbf w_n
\end{matrix} \right ].
\end{align*}
Let $Q=Q_1^{-1} Q_2$. Then, $Q$ is an orthogonal matrix and $\mathbf c _i Q=\mathbf w_i$ for $i\in \{1,2, \ldots, k\}$. Thus, $\C_1 Q=\C_2$.

If $\C_1, \C_2\in \mathrm{LCD}_{o,e}[n,k]$, $\C_1$ is an odd-like LCD code and $\C_1^{\perp}$ is even-like LCD code. By Theorem \ref{thm:odd-basis}, there is an orthonormal
basis $c_{1}, \ldots, c_{k}$ of $\C_1$. From Theorem \ref{thm:even-basis}, there is a basis $\mathbf c_{k+1},\mathbf c_{k+1}',  \ldots, \mathbf c_{k+\frac{n-k}{2}}, \mathbf c_{k+\frac{n-k}{2}}'$ of $\C_1^{\perp}$, which satisfies the conditions (i), (ii) and (iii) in Theorem \ref{thm:even-basis}. Similarly, $\C_2$ has an orthonormal
basis $\mathbf{w}_{1}, \ldots, \mathbf{w}_{k}$  and $\C_2^{\perp}$ has a basis $\mathbf w_{k+1},\mathbf w_{k+1}',  \ldots, \mathbf w_{k+\frac{n-k}{2}}, \mathbf w_{k+\frac{n-k}{2}}'$, which satisfies the conditions (i), (ii) and (iii) in Theorem \ref{thm:even-basis}. Let $Q_1$ and $Q_2$ be  matrices defined by
\begin{align*}
Q_1=\left [ \begin{matrix}
\mathbf c_1\\
\vdots \\
\mathbf c_k\\
\mathbf c_{k+1}\\
\mathbf c_{k+1}'\\
\vdots \\
\mathbf c_{k+\frac{n-k}{2}}\\
\mathbf c_{k+\frac{n-k}{2}}'
\end{matrix} \right ] \text{ and } Q_2=\left [ \begin{matrix}
\mathbf w_1\\
\vdots \\
\mathbf w_k\\
\mathbf w_{k+1}\\
\mathbf w_{k+1}'\\
\vdots \\
\mathbf w_{k+\frac{n-k}{2}}\\
\mathbf w_{k+\frac{n-k}{2}}'
\end{matrix} \right ].
\end{align*}
Then,

\begin{align*}
Q_1Q_1^T= Q_2 Q_2^T=
\left [ \begin{array}{cc}
I_k & 0\\
0 & \mathrm{diag}[\underbrace{J_2, \ldots, J_2}_{\frac{n-k}{2}}]
\end{array} \right ].
\end{align*}

Let $Q=Q_1^{-1} Q_2$.  One obtains $QQ^T=Q_1^{-1} (Q_2 Q_2^T) (Q_1^{-1})^T=Q_1^{-1} (Q_1 Q_1^T) (Q_1^{-1})^T=I_n$, that is, $Q$ is an orthogonal matrix.
It is observed that $\mathbf c_i Q=\mathbf w_i$ for every $i\in \{1, \ldots, k\}$. Thus, $\C_1 Q=\C_2$.

If $\C_1, \C_2\in \mathrm{LCD}_{e,o}[n,k]$, by a similar argument as the case $\C_1, \C_2\in \mathrm{LCD}_{o,e}[n,k]$,
one can prove that there is a  $Q\in \mathbb{O}_n$ such that $\C_2=\C_1 Q$. Hence, we prove that if $\C_1, \C_2$ have the same type, there is always an orthogonal matrix $Q$
such that $\C_2=\C_1 Q$.

Conversely, let $\C$ be an $[n,k]$ LCD code and $Q\in \mathbb O_n$. Recall that $\mathrm{wt}(\mathbf v) \pmod{2} \equiv \mathbf v \cdot \mathbf v = (\mathbf v Q) \cdot (\mathbf v Q)$
for any $\mathbf v\in \F_2^n$. Then,
$\C Q$ is odd-like (even-like, respectively) if and only if $\C$ is odd-like (even-like, respectively). Let $G$ be a generator matrix of $\C$. Then, $GQ$ is a generator matrix of $\C  Q$.
Note that $(GQ)(GQ)^T=GG^T$. Thus, $\C Q$ is LCD and its type is as the same as $\C$, which completes the proof.

\end{proof}

Theorem \ref{thm:C-Q} shows that the orthogonal group $\mathbb{O}_n$ acts on $\mathrm{LCD}[n,k]$ by $(\mathcal C,   Q) \longmapsto \C Q$, where $\C\in \mathrm{LCD}[n,k]$ and $Q\in \mathbb{O}_n$. The following theorem presents the decomposition of $\mathrm{LCD}[n,k]$ into $\mathbb{O}_n$-orbits.

\begin{theorem}
Let $k$ and $n$ be two positive integers such that $k< n$.

(i) If $n$ is odd  and $k$ is odd, $\mathrm{LCD}[n,k]$ can be decomposed as  the following  disjoint union of orbits
\begin{align*}
\mathrm{LCD}[n,k]=\C_{o,o} \mathbb O_n \cup \C_{o,e} \mathbb O_n.
\end{align*}

(ii) If $n$ is odd  and $k$ is even, $\mathrm{LCD}[n,k]$ can be decomposed as the following disjoint union of orbits
\begin{align*}
\mathrm{LCD}[n,k]=\C_{o,o} \mathbb O_n \cup \C_{e,o} \mathbb O_n.
\end{align*}

(iii)
  If $n$ is even  and $k$ is odd, $\mathrm{LCD}[n,k]$ can be decomposed as the following disjoint union of orbits
\begin{align*}
\mathrm{LCD}[n,k]=\C_{o,o} \mathbb O_n.
\end{align*}

(iv) If $n$ is even  and $k$ is even, $\mathrm{LCD}[n,k]$ can be decomposed as the following  disjoint union of orbits
\begin{align*}
\mathrm{LCD}[n,k]=\C_{o,o} \mathbb O_n \cup \C_{o,e} \mathbb O_n \cup \C_{e,o} \mathbb O_n.
\end{align*}
\end{theorem}
\begin{proof}
Firstly, recall that $\mathrm{LCD}_{o,e}[n,k]=\emptyset$ if $n-k$ is odd and $\mathrm{LCD}_{e,o}[n,k]=\emptyset$ if $k$ is odd.
Then, the results follow from Equation (\ref{eq:LCD-o-e}), Proposition \ref{prop:C-o-e} and Theorem \ref{thm:C-Q}.
\end{proof}

To determine the size of the $\mathbb O_n$-orbit of a code $\C$,  we need to  the order of the  stabilizer of $\C$, which is defined by
\begin{align*}
\mathrm{St}(\C)=\{Q\in \mathbb{O}_n: \C Q=\C\}.
\end{align*}
For an even positive integer, $\mathbb{S}\mathrm{p}_k$  denotes the symplectic group of degree $k$, a   $k\times k$  matrix $Q$ is in $\mathbb{S}\mathrm{p}_k$ if and only if
\begin{align*}
Q \mathrm{diag}[\underbrace{J_2, \ldots, J_2}_{\frac{k}{2}}]Q^T=\mathrm{diag}[\underbrace{J_2, \ldots, J_2}_{\frac{k}{2}}].
\end{align*}

\begin{lemma}\label{lem:St-C}
Let $\C$ be a binary $[n,k]$ LCD code, $G$ be a generator matrix of $\C$ and $H$ be a generator matrix of $\C^{\perp}$.
Then, $Q\in \mathrm{St}(\C)$ if and only if $Q= \left [ \begin{array}{c}
  G\\
  H
  \end{array} \right ]^{-1}
\left [ \begin{array}{cc}
 Q_1 & 0\\
 0 & Q_2
 \end{array} \right ] \left [ \begin{array}{c}
  G\\
  H
  \end{array} \right ],
$ where $Q_1\in \mathbb{GL}_k, Q_2\in \mathbb{GL}_{n-k}$ such that $Q_1 (G G^T) Q_1^T= G G^T$ and $Q_2 (H H^T)  Q_2^T= H H^T$.
\end{lemma}

\begin{proof}
Let $Q= \left [ \begin{array}{c}
  G\\
  H
  \end{array} \right ]^{-1}
\left [ \begin{array}{cc}
 Q_1 & 0\\
 0 & Q_2
 \end{array} \right ] \left [ \begin{array}{c}
  G\\
  H
  \end{array} \right ]$ satisfying $Q_1 (G G^T) Q_1^T= G G^T$ and $Q_2 (H H^T)  Q_2^T= H H^T$. Note that $\left [ \begin{array}{c}
  G\\
  H
  \end{array} \right ] Q= \left [ \begin{array}{cc}
 Q_1 & 0\\
 0 & Q_2
 \end{array} \right ] \left [ \begin{array}{c}
  G\\
  H
  \end{array} \right ]= \left [ \begin{array}{c}
   Q_1 G\\
  Q_2 H
  \end{array} \right ].$
  Then, $GQ=Q_1 G$ and $H Q =Q_2 H$. Thus, $\C Q= \C$.
  Note that
  \begin{align*}
   \left [ \begin{array}{c}
  G\\
  H
  \end{array} \right ] Q Q^T  \left [ \begin{array}{c}
  G\\
  H
  \end{array} \right ]^T=& \left ( \left [ \begin{array}{c}
  G\\
  H
  \end{array} \right ] Q\right )  \left (\left [ \begin{array}{c}
  G\\
  H
  \end{array} \right ] Q \right ) ^T\\
  =& \left [ \begin{array}{c}
   Q_1 G\\
  Q_2 H
  \end{array} \right ] \left [ \begin{array}{cc}
   G^TQ_1^T &   H^T Q_2^T
  \end{array} \right ]\\
  =& \left [ \begin{array}{cc}
   Q_1 G G^TQ_1^T &   Q_1G H^T Q_2^T\\
   Q_2 H G^TQ_1^T &   Q_2 H H^T Q_2^T
  \end{array} \right ]\\
  =& \left [ \begin{array}{cc}
   Q_1 G G^TQ_1^T &   0\\
   0 &   Q_2 H H^T Q_2^T
  \end{array} \right ]\\
  =& \left [ \begin{array}{cc}
    G G^T &   0\\
   0 &    H H^T
  \end{array} \right ]\\
  =& \left [ \begin{array}{c}
    G\\
   H
  \end{array} \right ] \left [ \begin{array}{c}
    G\\
   H
  \end{array} \right ]^T.
  \end{align*}
  Then, $QQ^T=I_n$, that is $Q\in \mathbb O_n$.
 Thus, $Q\in \mathrm{St}(\C)$.

  Conversely, let $Q\in \mathrm{St}(\C)$, that is $\C Q=\C$. Then, $GQ=Q_1 G$ with $Q_1\in \mathbb{GL}_k$. For any $\mathbf c\in \C$ and any $\mathbf w \in \C^{\perp}$,
  one has $(\mathbf c Q)\cdot (\mathbf w Q)= \mathbf c \cdot \mathbf w=0$ from $Q\in \mathbb O_n$. Then, $\C^{\perp} Q=\C^{\perp}$. Thus,
  there exists $Q_2 \in \mathbb{GL}_{n-k}$ such that $H Q=Q_2 H$. One gets $\left [ \begin{array}{c}
  G\\
  H
  \end{array} \right ] Q= \left [ \begin{array}{cc}
 Q_1 & 0\\
 0 & Q_2
 \end{array} \right ]  \left [ \begin{array}{c}
  G\\
  H
  \end{array} \right ]$. Then, $Q= \left [ \begin{array}{c}
  G\\
  H
  \end{array} \right ]^{-1}
\left [ \begin{array}{cc}
 Q_1 & 0\\
 0 & Q_2
 \end{array} \right ] \left [ \begin{array}{c}
  G\\
  H
  \end{array} \right ].
$
Since $Q\in \mathbb{O}_n$, one has
\begin{align*}
GG^T=&(GQ)(GQ)^T\\
=&Q_1(G G^T ) Q_1^T,
\end{align*}
and
\begin{align*}
HH^T=&(HQ)(HQ)^T\\
=&Q_2(H H^T ) Q_2^T.
\end{align*}
It completes the proof.
\end{proof}

\begin{corollary}\label{cor:St-C-o-o}
Let $k$ and $n$ be two positive integers with $0<k<n$.

(i) Let $\C_{o,o}$ be the LCD code  with generator matrix $G_{o,o}$ defined by Equation (\ref{eq:G-o-o}). Then,
$$\mathrm{St}(\C_{o,o})=\left \{\left [ \begin{array}{cc}
 Q_1 & 0\\
 0 & Q_2
 \end{array} \right ]: Q_1\in \mathbb O_k, Q_2 \in \mathbb O_{n-k} \right \}.$$

 (ii) Assume $(n-k)$ be even. Let $G_{o,e}$ and $H_{o,e}$  be  matrices defined by Equations (\ref{eq:G-o-o}) and (\ref{eq:H-o-o}) respectively. Then,
 $$\mathrm{St}(\C_{o,e})=\left \{ \left [ \begin{array}{c}
  G_{o,e}\\
  H_{o,e}
  \end{array} \right ]^{-1}  \left [ \begin{array}{cc}
 Q_1 & 0\\
 0 & Q_2
 \end{array} \right ] \left [ \begin{array}{c}
  G_{o,e}\\
  H_{o,e}
  \end{array} \right ]: Q_1\in \mathbb O_k, Q_2 \in \mathbb{S}\mathrm{p}_{n-k} \right \}.$$

   (iii) Assume $k$ be even. Let $G_{e,o}$ and $H_{e,o}$  be the matrices defined by Equations (\ref{eq:G-o-o}) and (\ref{eq:H-o-o}) respectively. Then,
 $$\mathrm{St}(\C_{e,o})=\left \{ \left [ \begin{array}{c}
  G_{e,o}\\
  H_{e,o}
  \end{array} \right ]^{-1}  \left [ \begin{array}{cc}
 Q_1 & 0\\
 0 & Q_2
 \end{array} \right ] \left [ \begin{array}{c}
  G_{e,o}\\
  H_{e,o}
  \end{array} \right ]: Q_1\in \mathbb{S}\mathrm{p}_{k}, Q_2 \in \mathbb O_{n-k}  \right \}.$$
\end{corollary}

\begin{proof}
Note that $G_{o,o}$, $G_{o,e}$ and $G_{e,o}$ are generator matrices of $\C_{o,o}$, $\C_{o,e}$ and $\C_{e,o}$ respectively, and
$H_{o,o}$, $H_{o,e}$ and $H_{e,o}$ are generator matrices of $\C_{o,o}$, $\C_{o,e}$ and $\C_{e,o}$ respectively.
For any $0<k<n$, $G_{o,o}G_{o,o}^T=I_k$ and $H_{o,o}H_{o,o}^T=I_{n-k}$. For even $(n-k)$, $G_{o,e}G_{o,e}^T=I_k$ and
$H_{o,e}H_{o,e}^T=\mathrm{diag}[\underbrace{J_2, \ldots, J_2}_{\frac{n-k}{2}}]$. For even $k$, $G_{e,o}G_{e,o}^T=\mathrm{diag}[\underbrace{J_2, \ldots, J_2}_{\frac{k}{2}}]$
and $H_{e,o}H_{e,o}^T=I_{n-k}$. Then, the results follow from Lemma \ref{lem:St-C}.
\end{proof}

For a finite set $S$, $|S|$ denotes the cardinality of $S$.
To determine the cardinality of the orbit, we need the following formulas, which can be found in \cite{Hum96} and \cite{Mac69}.
\begin{align}\label{eq:num-O}
|\mathbb O_k|= \begin{cases}
 2^{\frac{k^2}{4}}\prod_{i=1}^{\frac{k}{2}-1} (2^{2i}-1), & \text{ if } k \text{ is even,}
 \cr 2^{\frac{(k-1)^2}{4}}\prod_{i=1}^{\frac{k-1}{2}} (2^{2i}-1), & \text{ if } k \text{ is odd,}
 \end{cases}
\end{align}
and, for an even positive integer $k$,
\begin{align}\label{eq:num-Sp}
|\mathbb{S}\mathrm{p}_k|=
 2^{\frac{k^2}{4}}\prod_{i=1}^{\frac{k}{2}} (2^{2i}-1).
 \end{align}
For a number $q$ with $q\neq 1$, the Gaussian or $q$-binomial coefficient ${n \brack k}_q$ is defined to be
\begin{align*}
{n \brack k}_q=\frac{(q^{n}-1)(q^{n-2}-1)\cdots (q^{n-k+1}-1)}{(q-1)(q^2-1)\cdots (q^k-1)}.
\end{align*}
The  Gaussian coefficients has the same symmetry as that of binomial coefficients
$${n \brack k}_q= {n \brack n-k}_q.$$
The number of $k$-dimensional subspaces of an $n$-dimensional vector
space over $\F_q$ is just ${n \brack k}_q$.

\begin{theorem}\label{thm:size-orbits}
Let $k$ and $n$ be positive integers with $k<n$.

(i) Let $\mathrm{LCD}_{o,o}[n,k]$ be the set of odd-like binary $[n,k]$ LCD codes with odd-like duals, then
\begin{align*}
|\mathrm{LCD}_{o,o}[n,k]|= \begin{cases}
 2^{\frac{nk-k^2+n-1}{2}} {\frac{n}{2}-1 \brack \frac{k-1}{2}}_4, & \text{ if }  k \text{ odd}, n \text{ even,}
 \cr  2^{\frac{(n-k)(k-1)}{2}}   (2^{n-k}-1) {\frac{n-1}{2} \brack \frac{k-1}{2}}_4, & \text{ if }  k \text{ odd}, n \text{ odd,}
 \cr  2^{\frac{k(n-k-1)}{2}}   (2^{k}-1) {\frac{n-1}{2} \brack \frac{k}{2}}_4, & \text{ if }  k \text{ even}, n \text{ odd,}
 \cr  2^{\frac{k(n-k)}{2}}   (2^{k}-1) {\frac{n}{2}-1 \brack \frac{k}{2}}_4, & \text{ if }  k \text{ even}, n \text{ even.}
  \end{cases}
\end{align*}

(ii) Let $\mathrm{LCD}_{o,e}[n,k]$ be the set of odd-like binary $[n,k]$ LCD codes with even-like duals, then
\begin{align*}
|\mathrm{LCD}_{o,e}[n,k]|= \begin{cases}
 2^{\frac{(k-1)(n-k)}{2}} {\frac{n-1}{2} \brack \frac{k-1}{2}}_4, & \text{ if }  k \text{ odd}, n \text{ odd,}
 \cr  2^{\frac{k(n-k)}{2}}   {\frac{n}{2}-1 \brack \frac{k}{2}-1}_4, & \text{ if }  k \text{ even}, n \text{ even,}
\cr  0, & \text{ otherwise. }
 \end{cases}
\end{align*}

(iii) Let $\mathrm{LCD}_{e,o}[n,k]$ be the set of even-like binary $[n,k]$ LCD codes with odd-like duals, then
\begin{align*}
|\mathrm{LCD}_{e,o}[n,k]|= \begin{cases}
 2^{\frac{k(n-k-1)}{2}} {\frac{n-1}{2} \brack \frac{k}{2}}_4, & \text{ if }  k \text{ even}, n \text{ odd,}
 \cr  2^{\frac{k(n-k)}{2}}   {\frac{n}{2}-1 \brack \frac{k}{2}}_4, & \text{ if }  k \text{ even}, n \text{ even,}
\cr  0, & \text{ otherwise. }
 \end{cases}
\end{align*}
\end{theorem}
\begin{proof}
From Proposition \ref{prop:C-o-e} and Theorem \ref{thm:C-Q}, $\mathrm{LCD}_{o,o}= \C_{o,o} \mathbb O_n$. Then,
\begin{align*}
|\mathrm{LCD}_{o,o}|=\frac{| \mathbb O_n|}{|\mathrm{St}(\C_{o,o})|}.
\end{align*}
From Corollary \ref{cor:St-C-o-o}, one obtains
\begin{align*}
|\mathrm{LCD}_{o,o}|=\frac{| \mathbb O_n|}{| \mathbb O_k|\cdot | \mathbb O_{n-k}|}.
\end{align*}
Then, Part (i) follows from Equation (\ref{eq:num-O}).

By a similar discussion, one has
\begin{align*}
|\mathrm{LCD}_{o,e}|=\frac{| \mathbb O_n|}{| \mathbb O_k|\cdot | \mathbb{S}\mathrm{p}_{n-k}|}, \text{ if }  (n-k) \text{ is even, }
\end{align*}
and
\begin{align*}
|\mathrm{LCD}_{e,o}|=\frac{| \mathbb O_n|}{| \mathbb{S}\mathrm{p}_{k}|\cdot | \mathbb O_{n-k}|}, \text{ if }  k \text{ is even}.
\end{align*}
Then, Parts (ii) and (iii) follow from Equations (\ref{eq:num-O}) and (\ref{eq:num-Sp}).
It completes the proof.
\end{proof}

\begin{corollary}\label{cor:size-odd-even}
Let $k$ and $n$ be positive integers with $k<n$.

(i) Let $\mathrm{LCD}_{o}[n,k]$ be the set of odd-like binary $[n,k]$ LCD codes, then
\begin{align*}
|\mathrm{LCD}_{o}[n,k]|= \begin{cases}
 2^{\frac{nk-k^2+n-1}{2}} {\frac{n}{2}-1 \brack \frac{k-1}{2}}_4, & \text{ if }  k \text{ odd}, n \text{ even,}
 \cr  2^{\frac{(n-k)(k+1)}{2}}    {\frac{n-1}{2} \brack \frac{k-1}{2}}_4, & \text{ if }  k \text{ odd}, n \text{ odd,}
 \cr  2^{\frac{k(n-k-1)}{2}}   (2^{k}-1) {\frac{n-1}{2} \brack \frac{k}{2}}_4, & \text{ if }  k \text{ even}, n \text{ odd,}
 \cr  2^{\frac{(k+2)(n-k)}{2}}   {\frac{n}{2}-1 \brack \frac{k}{2}-1}_4, & \text{ if }  k \text{ even}, n \text{ even.}
  \end{cases}
\end{align*}

(ii)  Let $\mathrm{LCD}_{e}[n,k]$ be the set of even-like binary $[n,k]$ LCD codes, then
\begin{align*}
|\mathrm{LCD}_{e}[n,k]|= \begin{cases}
 2^{\frac{k(n-k-1)}{2}} {\frac{n-1}{2} \brack \frac{k}{2}}_4, & \text{ if }  k \text{ even}, n \text{ odd,}
 \cr  2^{\frac{k(n-k)}{2}}   {\frac{n}{2}-1 \brack \frac{k}{2}}_4, & \text{ if }  k \text{ even}, n \text{ even,}
\cr  0, & \text{ otherwise. }
 \end{cases}
\end{align*}
\end{corollary}

\begin{proof}
 From $|\mathrm{LCD}_{o}[n,k]|=|\mathrm{LCD}_{o,o}[n,k]|+|\mathrm{LCD}_{o,e}[n,k]|$ and Theorem  \ref{thm:size-orbits},
 Part (i) follows. From $|\mathrm{LCD}_{e}[n,k]|=|\mathrm{LCD}_{e,o}[n,k]|$ and Theorem  \ref{thm:size-orbits},
Part (ii) follows.
\end{proof}

\begin{corollary}
Let $k$ and $n$ be two positive integers with $k<n$. Then,
\begin{align*}
|\mathrm{LCD}[n,k]|= \begin{cases}
 2^{\frac{nk-k^2+n-1}{2}} {\frac{n}{2}-1 \brack \frac{k-1}{2}}_4, & \text{ if }  k \text{ odd}, n \text{ even,}
 \cr  2^{\frac{(n-k)(k+1)}{2}}    {\frac{n-1}{2} \brack \frac{k-1}{2}}_4, & \text{ if }  k \text{ odd}, n \text{ odd,}
 \cr  2^{\frac{k(n-k+1)}{2}}   {\frac{n-1}{2} \brack \frac{k}{2}}_4, & \text{ if }  k \text{ even}, n \text{ odd,}
 \cr  2^{\frac{k(n-k)}{2}} \left(2^{n-k} {\frac{n}{2}-1 \brack \frac{k}{2}-1}_4+      {\frac{n}{2}-1 \brack \frac{k}{2}}_4 \right ), & \text{ if }  k \text{ even}, n \text{ even.}
  \end{cases}
\end{align*}
\end{corollary}
\begin{proof}
From $|\mathrm{LCD}[n,k]|=|\mathrm{LCD}_{o}[n,k]|+|\mathrm{LCD}_{e}[n,k]|$ and Corollary \ref{cor:size-odd-even},
this corollary follows.
\end{proof}
\begin{remark}
In \cite{Sen97}, Sindrier gave a formula of the number of LCD codes, which involves the number of self-orthogonal codes. Since the number of terms in the summation of the formula is
very large, the formula is intractable. The formula we give here is a very simple closed formula.
\end{remark}

In the following, we shall  analyze the asymptotic behavior of the size of the orbits $\C \mathbb O_n$.
For all $q>1$, let $g_{q,n}$ be the number defined by
\begin{align*}
g_{q,n}=\prod_{i=1}^{n} (1-\frac{1}{q^i}).
\end{align*}
Then, we can rewrite ${n \brack k}_q$ as
\begin{align}\label{eq:g-q-n}
{n \brack k}_q=q^{k(n-k)} \frac{g_{q,n}}{g_{q,k}g_{q,n-k}}.
\end{align}

The sequence $g_{q,1}, g_{q,2}, \ldots$ is strictly decreasing and has positive limit, which is denoted by $g_{q, \infty}$ or $\prod_{i=1}^{\infty} (1-\frac{1}{q^i})$.
Then, one has
\begin{align}\label{eq:lim-bino}
\lim_{\substack{k \to \infty  \\ (n-k) \to \infty}} \frac{{n \brack k}_q}{q^{k(n-k)}}= \frac{1}{g_{q,\infty}}.
\end{align}

\begin{theorem}\label{thm:lim-LCD-o-o}
Let $k$ and $n$ be two positive integers with $k<n$.

(i) Let $\mathrm{LCD}_{o,o}[n,k]$ be the set of odd-like binary $[n,k]$ LCD codes with odd-like duals. Then,
\begin{align*}
 \lim_{\substack{k \to \infty  \\ (n-k) \to \infty}} \frac{|\mathrm{LCD}_{o,o}[n,k]|}{2^{k(n-k)}}= \frac{1}{g_{4,\infty}}.
 \end{align*}

(ii) Let $\mathrm{LCD}_{o,e}[n,k]$ be the set of odd-like binary $[n,k]$ LCD codes with even-like duals. If $(n-k)$ is odd, then $|\mathrm{LCD}_{o,e}[n,k]|=0$.
If $(n-k)$ is even, then
\begin{align*}
 \lim_{\substack{k \to \infty  \\ (n-k) \to \infty \\(n-k) \text{ is even }}} \frac{|\mathrm{LCD}_{o,e}[n,k]|}{2^{(k-1)(n-k)}}= \frac{1}{g_{4,\infty}}.
 \end{align*}

 (iii) Let $\mathrm{LCD}_{e,o}[n,k]$ be the set of even-like binary $[n,k]$ LCD codes with odd-like duals.
 If $k$ is odd, then $|\mathrm{LCD}_{e,o}[n,k]|=0$.
If $k$ is even, then
\begin{align*}
 \lim_{\substack{k \to \infty  \\ (n-k) \to \infty \\k \text{ is even }}} \frac{|\mathrm{LCD}_{e,o}[n,k]|}{2^{k(n-k-1)}}= \frac{1}{g_{4,\infty}}.
 \end{align*}
\end{theorem}
\begin{proof}
From Theorem \ref{thm:size-orbits} and Equation (\ref{eq:g-q-n}), one gets
\begin{align*}
|\mathrm{LCD}_{o,o}[n,k]|= \begin{cases}
 2^{k(n-k)}  \frac{g_{4, \frac{n}{2}-1}}{ g_{4, \frac{k-1}{2}}  g_{4, \frac{n-k-1}{2}} }, & \text{ if }  k \text{ odd}, n \text{ even,}
 \cr  2^{k(n-k)}   (1-2^{-(n-k)})  \frac{g_{4, \frac{n-1}{2}}}{ g_{4, \frac{k-1}{2}}  g_{4, \frac{n-k}{2}} }  , & \text{ if }  k \text{ odd}, n \text{ odd,}
 \cr  2^{k(n-k)}   (1-2^{-k})     \frac{g_{4, \frac{n-1}{2}}}{ g_{4, \frac{k}{2}}  g_{4, \frac{n-k-1}{2}} }  , & \text{ if }  k \text{ even}, n \text{ odd,}
 \cr  2^{k(n-k)}   (1-2^{-k})    \frac{g_{4, \frac{n}{2}-1}}{ g_{4, \frac{k}{2}}  g_{4, \frac{n-k}{2}} }  , & \text{ if }  k \text{ even}, n \text{ even.}
  \end{cases}
\end{align*}
From $\lim_{n \to \infty} g_{4,n}=g_{4,\infty}=\prod_{i=1}^{\infty} (1- \frac{1}{2^{2i}})$, Part (i) follows.
By a similar discussion, we can prove that Parts (ii) and (iii) hold. It completes the proof.
\end{proof}

\begin{corollary}\label{cor:LCD-o-o-codes}
 Let $k$ and $n$ be two positive integers with $k<n$.

(i) Let $\mathrm{LCD}_{o,o}[n,k]$ be the set of odd-like binary $[n,k]$ LCD codes with odd-like duals. Then,
\begin{align*}
 \lim_{\substack{k \to \infty  \\ (n-k) \to \infty}} \frac{|\mathrm{LCD}_{o,o}[n,k]|}{{n \brack k}_2}= \frac{1}{\prod_{i=1}^{\infty} (1+\frac{1}{2^i})}.
 \end{align*}

(ii) Let $\mathrm{LCD}_{o,e}[n,k]$ be the set of odd-like binary $[n,k]$ LCD codes with even-like duals. If $(n-k)$ is odd, then $|\mathrm{LCD}_{o,e}[n,k]|=0$.
If $(n-k)$ is even, then
\begin{align*}
 \lim_{\substack{k \to \infty  \\ (n-k) \to \infty \\(n-k) \text{ is even }}} \frac{2^{n-k}|\mathrm{LCD}_{o,e}[n,k]|}{{n \brack k}_2}= \frac{1}{\prod_{i=1}^{\infty} (1+\frac{1}{2^i})}.
 \end{align*}

 (iii) Let $\mathrm{LCD}_{e,o}[n,k]$ be the set of even-like binary $[n,k]$ LCD codes with odd-like duals.
 If $k$ is odd, then  $|\mathrm{LCD}_{e,o}[n,k]|=0$.
If $k$ is even,  then
\begin{align*}
 \lim_{\substack{k \to \infty  \\ (n-k) \to \infty \\k \text{ is even }}} \frac{2^k|\mathrm{LCD}_{e,o}[n,k]|}{{n \brack k}_2}= \frac{1}{\prod_{i=1}^{\infty} (1+\frac{1}{2^i})}.
 \end{align*}
\end{corollary}
\begin{proof}
From Theorem \ref{thm:lim-LCD-o-o} and Equation (\ref{eq:lim-bino}), one obtains
\begin{align*}
 \lim_{\substack{k \to \infty  \\ (n-k) \to \infty}} \frac{|\mathrm{LCD}_{o,o}[n,k]|}{{n \brack k}_2}=\frac{g_{2,\infty}}{g_{4,\infty}}.
 \end{align*}
It is observed that
\begin{align*}
\frac{g_{2,\infty}}{g_{4,\infty}}=&\frac{\lim_{n\to \infty} \prod_{i=1}^n (1-\frac{1}{2^i})}{\lim_{n\to \infty} \prod_{i=1}^n (1-\frac{1}{2^{2i}})}\\
=&\lim_{n\to \infty} \frac{ \prod_{i=1}^n (1-\frac{1}{2^i})}{ \prod_{i=1}^n (1-\frac{1}{2^{2i}})}\\
=&\lim_{n\to \infty} \frac{ 1}{ \prod_{i=1}^n (1+\frac{1}{2^{i}})}\\
=& \frac{ 1}{ \prod_{i=1}^{\infty} (1+\frac{1}{2^{i}})}.
\end{align*}
Then, Part (i) holds. By a similar discussion, we can prove Parts (ii) and (iii).
\end{proof}

\begin{corollary}\label{4.11}
Let $k$ and $n$ be two positive integers with $k<n$.

(i) Let $\mathrm{LCD}[n,k]$ be the set of  binary $[n,k]$ LCD codes. Then,
\begin{align*}
 \lim_{\substack{k \to \infty  \\ (n-k) \to \infty}} \frac{|\mathrm{LCD}[n,k]|}{{n \brack k}_2}= \frac{1}{\prod_{i=1}^{\infty} (1+\frac{1}{2^i})}.
 \end{align*}

(ii) Let $\mathrm{LCD}_{o,o}[n,k]$, $\mathrm{LCD}_{o,e}[n,k]$ and $\mathrm{LCD}_{e,o}[n,k]$ be defined as in Corollary \ref{cor:LCD-o-o-codes}. Then,
\begin{align*}
 \lim_{\substack{k \to \infty  \\ (n-k) \to \infty}} \frac{|\mathrm{LCD}_{o,o}[n,k]|}{|\mathrm{LCD}[n,k]|}= 1,
 \end{align*}
 and
 \begin{align*}
 \lim_{\substack{k \to \infty  \\ (n-k) \to \infty}} \frac{|\mathrm{LCD}_{o,e}[n,k]|}{|\mathrm{LCD}[n,k]|}= \frac{|\mathrm{LCD}_{e,o}[n,k]|}{|\mathrm{LCD}[n,k]|}=0.
 \end{align*}
\end{corollary}
\begin{proof}
From Corollary \ref{cor:LCD-o-o-codes},
this corollary follows.
\end{proof}

\begin{remark}
Part (i) of Corollary \ref{4.11} has been proved in \cite{Sen97}. We  give  a simpler and more direct proof for this result. Part (ii) shows that when $k$ and $(n-k)$  go to
infinity, almost all binary LCD codes $\C$ are odd-like codes with odd-like duals $\C^{\perp}$. Accordingly, for almost all binary LCD  codes,
they  and  their duals both have orthonormal basis.
\end{remark}

In the following,  we  count  the inequivalent binary LCD codes. Let $\mathbb P_n$ be the group generated by  all $n\times n$ permutation matrices, which  are  square  matrices that have  exactly one entry of $1$ in each row and each column and $0$s elsewhere. Two codes $\C_1$, $\C_2$ are equivalent if there exists a permutation $P$ in $\mathbb P_n$ such that
$\C_2=\C_1 P$. For every $\C \in \mathrm{LCD}[n,k]$ and every $P\in \mathbb P_n$, $\C$ and $\C P$ are in the same orbit. Then, we only need to classify LCD codes over every orbit.
We first consider classifying LCD codes in $\mathrm{LCD}_{o,o}[n,k]=\C_{o,o} \mathbb O_n$. For every $\C \in \mathrm{LCD}[n,k]$, $[\C]$ denotes the equivalence class of $\C$, i.e.,
$[\C]=\{\C P: P\in \mathbb P_n\}$. Let $\widetilde{\mathrm{LCD}}_{o,o}[n,k]=\{[\C]: \C \in \mathrm{LCD}_{o,o}[n,k]\}$. From $\mathrm{LCD}_{o,o}[n,k]= \C_{o,o} \mathbb O_n$, there is  a  one-to-one correspondence between
the family $\mathrm{LCD}_{o,o}[n,k]$ of LCD codes and the  the set  $\mathrm{St}(\C_{o,o}) \backslash \mathbb O_n$ of right cosets defined by
\begin{align*}
\pi: \mathrm{LCD}_{o,o}[n,k]   &\longrightarrow    \mathrm{St}(\C_{o,o}) \backslash \mathbb O_n,\\
\C_{o,o} Q      &\longmapsto          \mathrm{St}(\C_{o,o}) Q.
\end{align*}
By the definition of $\widetilde{\mathrm{LCD}}_{o,o}[n,k]$, the map $\pi$ induces a surjection $\tilde{\pi}$ between $\widetilde{\mathrm{LCD}}_{o,o}[n,k]$
and the set $\mathrm{St}(\C_{o,o}) \backslash \mathbb O_n / \mathbb{P}_n$ of double cosets defined as
\begin{align*}
\tilde{\pi}: \widetilde{\mathrm{LCD}}_{o,o}[n,k]   &\longrightarrow    \mathrm{St}(\C_{o,o}) \backslash \mathbb O_n  / \mathbb{P}_n,\\
[\C_{o,o} Q  ]    &\longmapsto          \mathrm{St}(\C_{o,o}) Q  \mathbb{P}_n .
\end{align*}
Then, we obtain a parametrization of the inequivalent linear codes in $\mathrm{LCD}_{o,o}[n,k]$
by the set $\mathrm{St}(\C_{o,o}) \backslash \mathbb O_n  / \mathbb{P}_n$ of double cosets.
Hence, classifying the inequivalent LCD codes in $\mathrm{LCD}_{o,o}[n,k]$ is
equivalent to determining representatives of the set $\mathrm{St}(\C_{o,o}) \backslash \mathbb O_n  / \mathbb{P}_n$ of double cosets.

For any $[n,k]$ linear code $\C$, the automorphism group $\mathrm{Aut}(\C)$ of $\C$ is defined by $\mathrm{Aut}(\C)=\{P\in \mathbb P_n:  \C P=\C\}$.
Then, we have the mass formula for $\mathrm{LCD}_{o,o}[n,k]$.
\begin{proposition}
Let $k$ and $n$ be two positive integers with $k<n$. Then,
\begin{align*}
\sum_{[\C] \in \widetilde{\mathrm{LCD}}_{o,o}[n,k]}  \frac{1}{|\mathrm{Aut}(\C)|}  = \frac{|\mathbb O_n|}{|\mathrm{St}(\C_{o,o})|\cdot  |\mathbb{P}_n|}.
\end{align*}
\end{proposition}
\begin{proof}
From $\mathrm{LCD}_{o,o}[n,k]=\C_{o,o} \mathbb O_n$, one has
\begin{align*}
 \frac{|\mathbb O_n|}{|\mathrm{St}(\C_{o,o})|}=&\sum_{\C \in \mathrm{LCD}_{o,o}[n,k]}  1 \\
 =& \sum_{[\C] \in \widetilde{\mathrm{LCD}}_{o,o}[n,k]}  |[\C]| \\
 =& \sum_{[\C] \in \widetilde{\mathrm{LCD}}_{o,o}[n,k]}  |\C \mathbb{P}_n| \\
 =& \sum_{[\C] \in \widetilde{\mathrm{LCD}}_{o,o}[n,k]}  \frac{| \mathbb{P}_n|}{ |\mathrm{Aut}(\C)|},
\end{align*}
which completes the proof.

\end{proof}

We can consider the problem of classifying the inequivalent LCD codes in $\mathrm{LCD}_{o,e}[n,k]$ and $\mathrm{LCD}_{e,o}[n,k]$
exactly like the case $\mathrm{LCD}_{o,o}[n,k]$. In short,  classifying the inequivalent binary LCD codes is
equivalent to determining representatives of the following three sets of double cosets:
\begin{align*}
\mathrm{St}(\C_{o,o}) \backslash \mathbb O_n  / \mathbb{P}_n,  \mathrm{St}(\C_{o,e}) \backslash \mathbb O_n  / \mathbb{P}_n
\end{align*}
and
\begin{align*}
\mathrm{St}(\C_{e,o}) \backslash \mathbb O_n  / \mathbb{P}_n.
\end{align*}

\section{The Characterization of LCD codes in odd characteristic}\label{sec:action-odd}
In this section, we will consider LCD codes over $\F_q$, where $q$ is a power of an odd prime.

The following proposition \cite{Ser12} is very important for the characterization of LCD codes over finite fields of odd characteristic.
\begin{proposition}\label{prop:diagonal-odd}
If $M$ is a $k\times k$ nonsingular symmetric matrix over $\F_q$ with $k\ge 2$,
then there is a $k\times k$ nonsingular matrix $Q$ such that $Q M Q^T=\mathrm{diag}[\underbrace{1, 1,  \ldots, 1}_{k-1}, \delta]$
, where $\delta=1$ if $\mathrm{det}(M)$ is a square  in $\F_q$, and  $\delta$ is any nonsquare in $\F_q$ if $\mathrm{det}(M)$ is a nonsquare in $\F_q$.
\end{proposition}
The following theorem presents a characterization of LCD codes in terms of their basis.
\begin{theorem}
Let $q$ be a power of an odd prime and $\C$ be an $[n,k]$ LCD code over $\F_q$. Then, $\C$ is LCD if and only if there is a generator matrix $G$ of $\C$ such that
$GG^T= \mathrm{diag}[\underbrace{1, 1,  \ldots, 1}_{k-1}, \delta]$, where $\delta\in \F_q^{*}=\F_q\setminus \{0\}$, that is, there is a basis $\mathbf e_1, \ldots,
 \mathbf e_k$ of $\C$ such that for any $i, j\in \{1,2, \ldots, k\}$,

 (i) $\mathbf e_i \cdot \mathbf e_j=0$ if $i\neq j$;

 (ii) $\mathbf e_i \cdot \mathbf e_i=1$ if $i\neq k$;

 (ii) $\mathbf e_k \cdot \mathbf e_k=\delta$.
\end{theorem}
\begin{proof}
Using Proposition \ref{prop:diagonal-odd}, we can prove this theorem by a similar discussion as in the proof of  Theorem \ref{thm:odd-basis} and Theorem \ref{thm:even-basis}.
\end{proof}

Let $G_1$ and $G_2$ be any two generator matrices of an $[n,k]$ code $\C$. Then, there is a  $k\times k$
nonsingular matrix $Q$ such that $G_2=Q G_1$. Thus, $\mathrm{det}(G_2G_2^T)=\mathrm{det}(Q)^2
\mathrm{det}(G_1G_1^T)$.  Then, $\mathrm{det}(G_2G_2^T) \mathrm{det}(G_1G_1^T)^{-1}\in (\F_q^{*})^2$.
Hence, we can define $\mathrm{LCD}_{+}[n,k]_q$ ($\mathrm{LCD}_{-}[n,k]_q$, respectively) be  the set of all $[n,k]$ LCD codes over $\F_q$ with $\mathrm{det}(GG^T)\in (\F_q^{*})^2$
($\mathrm{det}(GG^T)\not \in (\F_q^{*})^2$, respectively).
Let $\mathrm{LCD}[n,k]_q$ be the set of all $[n,k]$ LCD codes over $\F_q$. Then, $\mathrm{LCD}[n,k]_q=\mathrm{LCD}_{+}[n,k]_q \cup \mathrm{LCD}_{-}[n,k]_q $.
To construct linear codes in $\mathrm{LCD}_{-}[n,k]_q$, one need the following lemma \cite{Ser12}.
\begin{lemma}\label{eq:x^2+y^2}
For any $z\in \F_q$, there exist $x$ and $y$ in $\F_q$ such that $z=x^2+y^2$.
\end{lemma}
Let $\gamma$ be a nonsquare in $\F_q$. From Lemma \ref{eq:x^2+y^2}, there exist $a$ and $b$ in $\F_q$ such that $\gamma=a^2+b^2$.
 Let $G_{+}$ and $G_{-}$ be $k\times n$ matrices defined by
\begin{align}\label{eq:G-+}
G_{+}=\left [\begin{matrix}
\mathbf{e}_1\\
\mathbf{e}_2\\
\vdots\\
\mathbf{e}_{k-1}\\
\mathbf{e}_{k}
\end{matrix} \right ] \text{ and } G_{-}=\left [\begin{matrix}
\mathbf{e}_1\\
\mathbf{e}_2\\
\vdots\\
\mathbf{e}_{k-1}\\
a\mathbf{e}_{k}+b\mathbf{e}_{k+1}
\end{matrix} \right ].
\end{align}
Let $\C_{+}$ and $\C_{-}$ be linear codes generated by $G_{+}$ and $G_{-}$ respectively. Since $G_{+}G_{+}^T=I_{k}$
 and $G_{-}G_{-}^T=\mathrm{diag}[\underbrace{1, 1,  \ldots, 1}_{k-1}, \gamma]$, $\C_{+} \in \mathrm{LCD}_{+}[n,k]_q$ and $\C_{-} \in \mathrm{LCD}_{-}[n,k]_q$.
 Moreover, $\C_{+}$ and $\C_{-}$ have parity-check matrices $H_{+}$ and $H_{-}$ respectively, where
 \begin{align}\label{eq:H-+}
H_{+}=\left [\begin{matrix}
\mathbf{e}_{k+1}\\
\mathbf{e}_{k+2}\\
\vdots\\
\mathbf{e}_{n-1}\\
\mathbf{e}_{n}
\end{matrix} \right ] \text{ and } H_{-}=\left [\begin{matrix}
\mathbf{e}_{k+2}\\
\mathbf{e}_{k+3}\\
\vdots\\
\mathbf{e}_{n}\\
-b\mathbf{e}_{k}+a\mathbf{e}_{k+1}
\end{matrix} \right ].
\end{align}
Then, $\C_{+}^{\perp} \in \mathrm{LCD}_{+}[n,n-k]_q$ and $\C_{-}^{\perp} \in \mathrm{LCD}_{-}[n,n-k]_q$.
Further, we have the following results.
\begin{proposition}\label{prop:LCD+-}
Let $0<k<n$. If $\C\in \mathrm{LCD}_{+}[n,n-k]_q$ ($\C\in \mathrm{LCD}_{-}[n,n-k]_q$, respectively), then $\C^{\perp}\in \mathrm{LCD}_{+}[n,n-k]_q$ ($\C^{\perp}\in \mathrm{LCD}_{-}[n,n-k]_q$, respectively).
\end{proposition}
\begin{proof}
Let $G$ be a generator matrix of $\C$ and $H$ be a parity-check matrix of $\C$. Note that $GH^T=0$. Then,
\begin{align*}
\left (\mathrm{det}\left (\left [\begin{matrix}
G\\
H
\end{matrix} \right ]
\right ) \right )^2=&\mathrm{det}\left (\left [\begin{matrix}
G\\
H
\end{matrix} \right ]   \left [\begin{matrix}
G\\
H
\end{matrix} \right ]^T
\right ) \\
=& \mathrm{det}\left (\left [\begin{matrix}
G G^T & 0\\
0 & HH^T
\end{matrix} \right ]
\right ) \\
=& \mathrm{det}\left (   GG^T\right ) \mathrm{det}\left (   HH^T\right ).
\end{align*}
Then, $ \mathrm{det}\left (   GG^T\right ) \mathrm{det}\left (   HH^T\right ) \in (\F_q^*)^2$, which completes the proof.
\end{proof}

Let $\mathbb{O}_n(q)$ be the set of all $n\times n$ matrix $Q$ over $\F_q$ such that $QQ^T=I_n$. It is observed that $\mathbb{O}_n(q)$ acts on $\mathrm{LCD}[n,n-k]_q$
by  $(\mathcal C,   Q) \longmapsto \C Q$, where $\C\in \mathrm{LCD}[n,k]_q$ and $Q\in \mathbb{O}_n(q)$.
Moreover, $\mathrm{LCD}_{+}[n,k]_q$ and $\mathrm{LCD}_{-}[n,k]_q$ are $\mathbb{O}_n(q)$-invariant subsets. In fact, we have the following stronger results.
\begin{proposition}
Let $k$ and $n$ be two positive integers with $k<n$. Then, $\mathrm{LCD}_{+}[n,k]_q= \C_{+} \mathbb{O}_n(q)$ and $\mathrm{LCD}_{-}[n,k]_q= \C_{-} \mathbb{O}_n(q)$.
Hence, $\mathrm{LCD}[n,k]_q= \C_{+} \mathbb{O}_n(q) \cup \C_{-} \mathbb{O}_n(q)$.
\end{proposition}
\begin{proof}
From $\mathrm{LCD}[n,k]_q= \mathrm{LCD}_{+}[n,k]_q \cup \mathrm{LCD}_{-}[n,k]_q$ and Proposition \ref{prop:LCD+-}, this proposition follows.
\end{proof}

\begin{lemma}\label{lem:St-C-odd}
Let $\C$ be an $[n,k]$ LCD code  over $\F_q$, $G$ be a generator matrix of $\C$ and $H$ be a generator matrix of $\C^{\perp}$.
Then, $Q\in \mathrm{St}(\C)$ if and only if $Q= \left [ \begin{array}{c}
  G\\
  H
  \end{array} \right ]^{-1}
\left [ \begin{array}{cc}
 Q_1 & 0\\
 0 & Q_2
 \end{array} \right ] \left [ \begin{array}{c}
  G\\
  H
  \end{array} \right ],
$ where $Q_1\in \mathbb{GL}_k(q)$ and $Q_2\in \mathbb{GL}_{n-k}(q)$ such that $Q_1 (G G^T) Q_1^T= G G^T$ and $Q_2 (H H^T)  Q_2^T= H H^T$.
\end{lemma}
\begin{proof}
The proof is analogous to the proof of Lemma \ref{lem:St-C}.
\end{proof}

For $\delta\in \F_q^*$, let $\mathbb{O}_k^{\delta}(q)$ be the group  defined by
 \begin{align*}
 \mathbb{O}_k^{\delta}(q)=\{Q\in \mathbb{O}_k(q): Q \mathrm{diag}[\underbrace{1, 1,  \ldots, 1}_{k-1}, \delta] Q^T= \mathrm{diag}[\underbrace{1, 1,  \ldots, 1}_{k-1}, \delta]\}.
 \end{align*}
Then, $\mathbb{O}_k^{1}(q)$ is just $\mathbb{O}_k(q)$.

\begin{corollary}\label{cor:St-C-+}
Let $k$ and $n$ be two positive integers with $0<k<n$.

(i) Let $\C_{+}$ be the LCD code  with 
 the generator matrix $G_{+}$ defined by Equation (\ref{eq:G-+}). Then,
$$\mathrm{St}(\C_{+})=\left \{\left [ \begin{array}{cc}
 Q_1 & 0\\
 0 & Q_2
 \end{array} \right ]: Q_1\in \mathbb O_k(q), Q_2 \in \mathbb O_{n-k}(q) \right \}.$$

 (ii) Let  $G_{-}$ and $H_{-}$  be  matrices defined by Equations (\ref{eq:G-+}) and (\ref{eq:H-+}). Then,
 $$\mathrm{St}(\C_{-})=\left \{ \left [ \begin{array}{c}
  G_{-}\\
  H_{-}
  \end{array} \right ]^{-1}  \left [ \begin{array}{cc}
 Q_1 & 0\\
 0 & Q_2
 \end{array} \right ] \left [ \begin{array}{c}
  G_{-}\\
  H_{-}
  \end{array} \right ]: Q_1\in \mathbb O_k^{\gamma}(q), Q_2 \in \mathbb O_{n-k}^{\gamma}(q) \right \}.$$

 \end{corollary}

 \begin{proof}
 This corollary follows from Lemma \ref{lem:St-C-odd}.
 \end{proof}

To determine the cardinality of the orbit $\C \mathbb{O}_n(q)$, we need the cardinality of $\mathbb{O}_n^{\delta}(q)$, which can be found in \cite{She,Weyl16}.
If $n$ is odd, one has
\begin{align}\label{eq:num-O-q-odd}
|\mathbb{O}_n^{\delta}(q)|=   2 q^{\frac{(n-1)^2}{4}}   \prod_{i=1}^{\frac{n-1}{2}} (q^{2i}-1).
\end{align}
If $n$ is even, one has
\begin{align}\label{eq:num-O-q-even}
|\mathbb{O}_n^{\delta}(q)|=  2 q^{\frac{n(n-2)}{4}}  \left (q^{\frac{n}{2}}-\eta \left ((-1)^{\frac{n}{2}}\delta \right ) \right )  \prod_{i=1}^{\frac{n}{2}-1} (q^{2i}-1),
\end{align}
where $\eta$ is the  Legendre character of $\F_q$.

\begin{theorem}\label{thm:size-orbits-q}
Let $q$ be a power of odd prime and  $k$, $n$ be 
two positive integers with $k<n$.

(i) $\mathrm{LCD}_{+}[n,k]_q = \C_{+} \mathbb{O}_n(q)$ and 
\begin{align*}
|\mathrm{LCD}_{+}[n,k]_q|= \begin{cases}
 \frac{1}{2}  q^{\frac{k(n-k)-1}{2}}   \left (q^{\frac{n}{2}}- \eta \left( (-1)^{\frac{n}{2}}    \right )  \right )             {\frac{n}{2}-1 \brack \frac{k-1}{2}}_{q^2}, & \text{ if }  k \text{ odd}, n \text{ even,}
 \cr  \frac{1}{2}  q^{\frac{k(n-k)}{2}}   \left (q^{\frac{n-k}{2}}+ \eta \left( (-1)^{\frac{n-k}{2}}    \right )  \right )             {\frac{n-1}{2} \brack \frac{k-1}{2}}_{q^2}, & \text{ if }  k \text{ odd}, n \text{ odd,}
 \cr  \frac{1}{2}  q^{\frac{k(n-k)}{2}}   \left (q^{\frac{k}{2}}+ \eta \left( (-1)^{\frac{k}{2}}    \right )  \right )  {\frac{n-1}{2} \brack \frac{k}{2}}_{q^2}, & \text{ if }  k \text{ even}, n \text{ odd,}
 \cr  \frac{1}{2}  q^{\frac{k(n-k)}{2}}    \frac{\left (q^{\frac{k}{2}}+ \eta \left( (-1)^{\frac{k}{2}}    \right )  \right )
 \left (q^{\frac{n-k}{2}}+ \eta \left( (-1)^{\frac{n-k}{2}}    \right )  \right )  }{\left (q^{\frac{n}{2}}+ \eta \left( (-1)^{\frac{n}{2}}    \right )  \right ) }
     {\frac{n}{2} \brack \frac{k}{2}}_{q^2}, & \text{ if }  k \text{ even}, n \text{ even.}
  \end{cases}
\end{align*}

(ii)  $\mathrm{LCD}_{-}[n,k]_q = \C_{-} \mathbb{O}_n(q)$ and 
\begin{align*}
|\mathrm{LCD}_{-}[n,k]_q|= \begin{cases}
 \frac{1}{2}  q^{\frac{k(n-k)-1}{2}}   \left (q^{\frac{n}{2}}- \eta \left( (-1)^{\frac{n}{2}}    \right )  \right )             {\frac{n}{2}-1 \brack \frac{k-1}{2}}_{q^2}, & \text{ if }  k \text{ odd}, n \text{ even,}
 \cr  \frac{1}{2}  q^{\frac{k(n-k)}{2}}   \left (q^{\frac{n-k}{2}}- \eta \left( (-1)^{\frac{n-k}{2}}    \right )  \right )             {\frac{n-1}{2} \brack \frac{k-1}{2}}_{q^2}, & \text{ if }  k \text{ odd}, n \text{ odd,}
 \cr  \frac{1}{2}  q^{\frac{k(n-k)}{2}}   \left (q^{\frac{k}{2}}- \eta \left( (-1)^{\frac{k}{2}}    \right )  \right )  {\frac{n-1}{2} \brack \frac{k}{2}}_{q^2}, & \text{ if }  k \text{ even}, n \text{ odd,}
 \cr  \frac{1}{2}  q^{\frac{k(n-k)}{2}}    \frac{\left (q^{\frac{k}{2}}- \eta \left( (-1)^{\frac{k}{2}}    \right )  \right )
 \left (q^{\frac{n-k}{2}}- \eta \left( (-1)^{\frac{n-k}{2}}    \right )  \right )  }{\left (q^{\frac{n}{2}}+ \eta \left( (-1)^{\frac{n}{2}}    \right )  \right ) }
     {\frac{n}{2} \brack \frac{k}{2}}_{q^2}, & \text{ if }  k \text{ even}, n \text{ even.}
  \end{cases}
\end{align*}

\end{theorem}
\begin{proof}
From $\mathrm{LCD}_{+}[n,k]_q = \C_{+} \mathbb{O}_n(q)$ and Corollary \ref{cor:St-C-+}, one has
\begin{align*}
|\mathrm{LCD}_{+}[n,k]_q| =& \frac{|\mathbb{O}_n(q)|}{|\mathrm{St}(\C_{+})|}\\
=& \frac{|\mathbb{O}_n(q)|}{|\mathbb{O}_{k}(q)| \cdot |\mathbb{O}_{n-k}(q)|}.
\end{align*}
Then, Part (i) follows from Equations (\ref{eq:num-O-q-odd}) and (\ref{eq:num-O-q-even}).

From $\mathrm{LCD}_{-}[n,k]_q = \C_{-} \mathbb{O}_n(q)$ and Corollary \ref{cor:St-C-+}, one has
\begin{align*}
|\mathrm{LCD}_{-}[n,k]_q| =& \frac{|\mathbb{O}_n(q)|}{|\mathrm{St}(\C_{-})|}\\
=& \frac{|\mathbb{O}_n(q)|}{|\mathbb{O}_{k}^{\gamma}(q)| \cdot |\mathbb{O}_{n-k}^{\gamma}(q)|}.
\end{align*}
Then, Part (ii) follows from Equations (\ref{eq:num-O-q-odd}) and (\ref{eq:num-O-q-even}).

\end{proof}

\begin{corollary}\label{cor:size-LCD-q}
Let $q$ be a power of an odd prime and  $k$, $n$ be two positive integers with $k<n$. Then
\begin{align*}
|\mathrm{LCD}[n,k]_q|= \begin{cases}
  q^{\frac{k(n-k)-1}{2}}   \left (q^{\frac{n}{2}}- \eta \left( (-1)^{\frac{n}{2}}    \right )  \right )             {\frac{n}{2}-1 \brack \frac{k-1}{2}}_{q^2}, & \text{ if }  k \text{ odd}, n \text{ even,}
 \cr    q^{\frac{(k+1)(n-k)}{2}}      {\frac{n-1}{2} \brack \frac{k-1}{2}}_{q^2}, & \text{ if }  k \text{ odd}, n \text{ odd,}
 \cr    q^{\frac{k(n-k+1)}{2}}     {\frac{n-1}{2} \brack \frac{k}{2}}_{q^2}, & \text{ if }  k \text{ even}, n \text{ odd,}
 \cr   q^{\frac{k(n-k)}{2}}     {\frac{n}{2} \brack \frac{k}{2}}_{q^2}, & \text{ if }  k \text{ even}, n \text{ even.}
  \end{cases}
\end{align*}

\end{corollary}

\begin{proof}
This corollary follows from Theorem \ref{thm:size-orbits-q}.
\end{proof}

\begin{corollary} \label{cor:ratio-LCD-+}
Let $k$ and $n$ be positive integers with $k<n$.

(i)  $\mathrm{LCD}_{+}[n,k]_q= \C_{+} \mathbb{O}_n(q)$ and 
\begin{align*}
 \lim_{\substack{k \to \infty  \\ (n-k) \to \infty}} \frac{|\mathrm{LCD}_{+}[n,k]_q|}{{n \brack k}_{q}}= \frac{1}{ 2 \prod_{i=1}^{\infty} (1+\frac{1}{q^i})}.
 \end{align*}

 (ii)  $\mathrm{LCD}_{-}[n,k]_q= \C_{-} \mathbb{O}_n(q)$ and 
\begin{align*}
 \lim_{\substack{k \to \infty  \\ (n-k) \to \infty}} \frac{|\mathrm{LCD}_{-}[n,k]_q|}{{n \brack k}_{q}}= \frac{1}{ 2 \prod_{i=1}^{\infty} (1+\frac{1}{q^i})}.
 \end{align*}

 (iii) Let $\mathrm{LCD}[n,k]_q$ be the set of all $[n,k]$ LCD codes over $\F_q$. Then,
\begin{align*}
 \lim_{\substack{k \to \infty  \\ (n-k) \to \infty}} \frac{|\mathrm{LCD}[n,k]_q|}{{n \brack k}_{q}}= \frac{1}{ \prod_{i=1}^{\infty} (1+\frac{1}{q^i})}.
 \end{align*}

\end{corollary}
\begin{proof}
From Theorem \ref{thm:size-orbits-q} and Equation (\ref{eq:g-q-n}), it completes the proof.
\end{proof}

\begin{corollary}
Let $k$ and $n$ be two positive integers with $k<n$. Then,
\begin{align*}
 \lim_{\substack{k \to \infty  \\ (n-k) \to \infty}} \frac{|\mathrm{LCD}_{+}[n,k]_q|}{|\mathrm{LCD}[n,k]_q|}= \lim_{\substack{k \to \infty  \\ (n-k) \to \infty}} \frac{|\mathrm{LCD}_{-}[n,k]_q|}{|\mathrm{LCD}[n,k]_q|}=\frac{1}{ 2 }.
 \end{align*}
Thus,
\begin{align*}
 \lim_{\substack{k \to \infty  \\ (n-k) \to \infty}} \frac{|\mathrm{LCD}_{+}[n,k]_q|}{|\mathrm{LCD}_{-}[n,k]_q|}= 1.
  \end{align*}

\end{corollary}

\begin{proof}
This corollary follows from Corollary \ref{cor:ratio-LCD-+}.
\end{proof}

By a similar discussion as the binary case, classifying the inequivalent $q$-array LCD codes is
equivalent to determining representatives of the following two sets of double cosets:
\begin{align*}
\mathrm{St}(\C_{+}) \backslash \mathbb O_n(q)  / \mathbb{P}_n    \text{ and }  \mathrm{St}(\C_{-}) \backslash \mathbb O_n(q)  / \mathbb{P}_n.
\end{align*}

\section{Concluding remarks}
In this paper, we have pushed further the study of the general structure of LCD codes. Firstly, we have provided a new characterization of LCD codes by their basis.  As a consequence, a conjecture on minimum distance of binary LCD codes proposed by Galvez et al. \cite{GKLRW17} was solved. Then, we have considered the action
 of the orthogonal group on the set of all LCD codes. All the possible orbits of this action have been identified  and  closed formulas of the
 size of the orbits have been derived. Our results show that almost all binary LCD codes are odd-like codes with odd-like duals  
 and about half of $q$-LCD codes have orthonormal basis, where $q$ is a power of an odd prime.


\begin{thebibliography}{10}

\bibitem{BCCGM14} J. Bringer, C. Carlet, H. Chabanne, S. Guilley, and H. Maghrebi. Orthogonal direct sum masking a smartcard friendly computation paradigm in a code, with builtin protection against side-channel and fault attacks. In WISTP, volume 8501 of Lecture Notes in Comput. Sci., pages 40-56. Springer, Berlin, 2014.


\bibitem{CG14} C. Carlet and S. Guilley. Complementary dual codes for counter- measures to side-channel attacks. In Proceedings of the 4th ICMCTA Meeting, volume 3 of CIM Series in Mathematical Sciences book series, pages 87-95. Springer, Berlin, 2014.






 \bibitem{CMTQ17} C. Carlet, S. Mesnager, C. Tang and Y. Qi, ¡°Linear codes over $\F_q$ which are equivalent to LCD codes,¡± arXiv preprint arXiv:1703.04346, 2017.

 \bibitem{CMTQ17-sigma} C. Carlet, S. Mesnager, C. Tang and Y. Qi, ¡°On $\sigma$-LCD codes ,¡± arXiv preprint arXiv:1707.08789, 2017.


\bibitem{DLL16}  C. Li, C. Ding and S. Li. LCD Cyclic codes over finite fields. IEEE Trans. Inf. Theory, vol. 63, no. 7, pp. 4344 - 4356, 2017.

\bibitem{DKOSS17} S. T. Dougherty, J. L. Kim, B. Ozkaya,  L. Sok and  P. Sol\'e. The combinatorics of LCD codes:
Linear Programming bound and orthogonal matrices. International Journal of Information and Coding Theory, 4(2-3), 116-128, 2017.

\bibitem{GKLRW17} L. Galvez, J. L. Kim, N. Lee, Y. G. Roe, and B. S. Won. Some Bounds on Binary LCD Codes. arXiv preprint arXiv:1701.04165, 2017.


\bibitem{Hum96} J. F. Humphreys. A course in group theory. Oxford Science Publications, The Clarendon Press, Oxford
University Press, New York, 1996.


\bibitem{LDL16} S. Li, C. Li, C. Ding, and H. Liu,
 Two Families of LCD BCH Codes.   IEEE Trans. Inf. Theory, vol 63, no. 9, pp. 5699 - 5717, 2017.


  \bibitem{MTQ16} S. Mesnager, C. Tang and Y. Qi, Complementary dual algebraic geometry codes, arXiv preprint arXiv:1609.05649, 2016.

  \bibitem{Mas92} James L. Massey. Linear codes with complementary duals. Discrete Mathematics, 106-107:337-342, 1992.

\bibitem{Mac69} J. MacWilliams.  Orthogonal matrices over finite fields. Am. Math. Monthly 76, 152-164 (1969).

\bibitem{She} Sheekey, J.; On rank problems for subspaces of matrices over finite fields, Ph.D. thesis.

\bibitem{Sen97} N. Sendrier. On the dimension of the hull. SIAM Journal on Discrete Mathematics, 10(2), 282-293, 1997.

\bibitem{Ser12} J. P. Serre. A course in arithmetic. Springer Science \& Business Media, 2012.

\bibitem{Weyl16} H. Weyl. The classical groups: their invariants and representations. Princeton university press, 2016.


\end{thebibliography}
\end{document}